\documentclass[12pt,draftclsnofoot,onecolumn]{IEEEtran}
\usepackage{cite}

\usepackage{graphicx}
\usepackage{subfigure}
\usepackage{epstopdf}
\usepackage{mathrsfs}
\usepackage{pifont}
\usepackage{multirow}
\usepackage{algorithm}
\usepackage{algorithmic}
\usepackage{amsthm,amsmath,amsfonts}
\usepackage{bbm}
\usepackage{mathrsfs}
\newtheorem{lemma}{\textbf{Lemma}}
\usepackage{bm}
\usepackage{subfigure}

\newtheorem{assumption}{Assumption}
\usepackage{color}

\begin{document}
%
\title{Feeling of Presence Maximization: mmWave-Enabled Virtual Reality Meets Deep Reinforcement Learning}
%
%
%

\author{Peng Yang,~\IEEEmembership{Member,~IEEE}, Tony Q. S. Quek,~\IEEEmembership{Fellow,~IEEE}, Jingxuan Chen, Chaoqun You,~\IEEEmembership{Member,~IEEE}, and Xianbin Cao,~\IEEEmembership{Senior Member,~IEEE}
\thanks{
P. Yang, T. Q. S. Quek, and C. You are with the Information Systems Technology and Design, Singapore University of Technology and Design, 487372 Singapore.
J. Chen and X. Cao are with the School of Electronic and Information Engineering, Beihang University, Beijing 100083, China.
}
}
\maketitle

\begin{abstract}
This paper investigates the problem of providing ultra-reliable and energy-efficient virtual reality (VR) experiences for wireless mobile users. To ensure reliable ultra-high-definition (UHD) video frame delivery to mobile users and enhance their immersive visual experiences, a coordinated multipoint (CoMP) transmission technique and millimeter wave (mmWave) communications are exploited. Owing to user movement and time-varying wireless channels, the wireless VR experience enhancement problem is formulated as a sequence-dependent and mixed-integer problem with a goal of maximizing users' feeling of presence (FoP) in the virtual world, subject to power consumption constraints on access points (APs) and users' head-mounted displays (HMDs).
The problem, however, is hard to be directly solved due to the lack of users' accurate tracking information and the sequence-dependent and mixed-integer characteristics. To overcome this challenge, we develop a parallel echo state network (ESN) learning method to predict users' tracking information by training fresh and historical tracking samples separately collected by APs.
With the learnt results, we propose a deep reinforcement learning (DRL) based optimization algorithm to solve the formulated problem. In this algorithm, we implement deep neural networks (DNNs) as a scalable solution to produce integer decision variables and solve a continuous power control problem to criticize the integer decision variables.
Finally, the performance of the proposed algorithm is compared with various benchmark algorithms, and the impact of different design parameters is also discussed.
Simulation results demonstrate that the proposed algorithm is more $4.14\%$ energy-efficient than the benchmark algorithms.
\end{abstract}


%
\IEEEpeerreviewmaketitle

\section{Introduction}
%
%
%
%
\IEEEPARstart{V}{irtual} reality (VR) applications have attracted tremendous interest in various fields, including entertainment, education, manufacturing, transportation, healthcare, and many other consumer-oriented services \cite{DBLP:journals/tmm/HouDZB21}. These applications exhibit enormous potential in the next generation of multimedia content envisioned by enterprises and consumers through providing richer and more engaging, and immersive experiences. According to market research \cite{Heather}, the VR ecosystem is predicted to be an $80$ billion market by 2025, roughly the size of the desktop PC market today.

However, several major challenges need to be overcome such that businesses and consumers can get fully on board with VR technology \cite{WiltzMajor}, one of which is to provide compelling content. To this aim, the resolution of provided content must be guaranteed.
In VR applications, VR wearers can either view objects up close or across a wide field of view (FoV) via head-mounted or goggle-type displays (HMDs).
As a result, very subtle defects such as poorly rendering pixels at any point on an HMD may be observed by a user up close, which may degrade users' truly visual experiences.
To create visually realistic images across the HMD, it must have more display pixels per eye, which indicates that ultra-high-definition (UHD) video frame transmission must be enabled for VR applications.
However, the transmission of UHD video frames typically requires $4-5$ times the system bandwidth occupied for delivering a regular high-definition (HD) video \cite{DBLP:journals/tmm/LiuLAC19,DBLP:journals/tcsv/DaiZML20}.
Further, to achieve good user visual experiences, the motion-to-photon latency should be ultra-low (e.g., $10-25$ ms) \cite{DBLP:journals/tmc/LaiHCSDL20,DBLP:conf/icccn/HouLD17,Qualcomm}. High motion-to-photon values will send conflicting signals to the Vestibulo-ocular reflex (VOR) and then might cause dizziness or motion sickness.

Hence, today's high-end VR systems such as Oculus Rift \cite{Oculus} and HTC Vive \cite{Htc} that offer high quality and accurate positional tracking remain tethered to deliver UHD VR video frames while satisfying the stringent low-latency requirement.
Nevertheless, wired VR display may degrade users' seamless visual experiences due to the constraint on the movement of users. Besides, a tethered VR headset presents a potential tripping hazard for users.
Therefore, to provide ultimate VR experiences, VR systems or at least the headset component should be untethered \cite{DBLP:journals/tmc/LaiHCSDL20}.

Recently, the investigation on wireless VR has attracted numerous attention from both industry and academe; of particular interest is how to a) develop mobile (wireless and lightweight) HMDs, b) how to enable seamless and immersive VR experiences on mobile HMDs in a bandwidth-efficiency manner, while satisfying ultra-low-latency requirements.

\subsection{Related work}
On the aspect of designing lightweight VR HMDs, considering heavy image processing tasks, which are usually insufficient in the graphics processing unit (GPU) of a local HMD, one might be persuaded to transfer the image processing from the local HMD to a cloud or network edge units (e.g., edge servers, base stations, and access points (APs)).
For example, the work in \cite{DBLP:journals/tmm/HouDZB21} proposed to enable mobile VR with lightweight VR glasses by completing computation-intensive tasks (such as encoding and rendering) on a cloud/edge server and then delivering video streams to users.
The framework of fog radio access networks, which could significantly relieve the computation burden by taking full advantages of the edge fog computing, was explored in \cite{DBLP:journals/jsac/DangP19} to facilitate the lightweight HMD design.

In terms of proposing VR solutions with improved bandwidth utilization, current studies can be classified into two groups: tiling and video coding \cite{DBLP:conf/hotnets/LiuXGH0V17}
As for tiling, some VR solutions propose to spatially divide VR video frames into small parts called tiles, and only tiles within users' FoV are delivered to users \cite{DBLP:conf/sigcomm/JuHSLLZH17,DBLP:conf/sigcomm/MangianteKNGRS17,DBLP:journals/tmm/GaddamREGH16}. The FoV of a user is defined as the extent of the observable environment at any given time. By sending HD tiles in users' FoV, the bandwidth utilization is improved.
On the aspect of video coding, the VR video is encoded into multiple versions of different quality levels. Viewers receive appropriate versions based on their viewing directions \cite{DBLP:conf/icc/CorbillonSDC17}.

Summarily, to improve bandwidth utilization, the aforementioned works \cite{DBLP:conf/sigcomm/JuHSLLZH17,DBLP:conf/sigcomm/MangianteKNGRS17,DBLP:journals/tmm/GaddamREGH16,DBLP:conf/icc/CorbillonSDC17} either transmit relatively narrow user FoV or deliver HD video frames.
Nevertheless, wider FoV is significantly important for a user to have immersive and presence experiences. Meanwhile, transmitting UHD video frames can enhance users' visual experiences.
To this aim, advanced wireless communication techniques (particularly, millimeter wave (mmWave)), which can significantly improve data rates and reduce propagation latency via providing wide bandwidth transmission, are explored in VR video transmission \cite{DBLP:journals/tmm/LiuLAC19,DBLP:journals/tcom/PerfectoESB20,DBLP:conf/wcnc/ElBambyPBD18}.
For example, the work in \cite{DBLP:journals/tmm/LiuLAC19} utilized a mmWave-enabled communication architecture to support the panoramic and UHD VR video transmission.
Aiming to improve users' immersive VR experiences in a wireless multi-user VR network, a mmWave multicast transmission framework was developed in \cite{DBLP:journals/tcom/PerfectoESB20}.
Besides, the mmWave communication for ultra-reliable and low latency wireless VR was investigated in \cite{DBLP:conf/wcnc/ElBambyPBD18}.

\subsection{Motivation and contributions}
Although mmWave techniques can alleviate the current bottleneck for UHD video delivery, mmWave links are prone to outage as they require line-of-sight (LoS) propagation.
Various physical obstacles in the environment (including users' bodies) may completely break mmWave links \cite{DBLP:journals/twc/ChenSSLY20}. As a result, VR requirements for a perceptible image-quality degradation-free uniform experience cannot be accommodated.
However, the mmWave VR-related works in \cite{DBLP:journals/tmm/LiuLAC19,DBLP:journals/tcom/PerfectoESB20,DBLP:conf/wcnc/ElBambyPBD18} did not effectively investigate the crucial issue of guaranteeing the transmission reliability of VR video frames.
To significantly improve the transmission reliability of VR video frames under low-latency constraints, the coordinated multipoint (CoMP) transmission technique, which can improve the reliability via spatial diversity, can be explored \cite{DBLP:journals/corr/abs-2002-09194}.
Besides, it is extensively considered that proactive computing (e.g., image processing or frame rendering) enabled by adopting machine learning methods is a crucial ability for a wireless VR network to mandate the stringent low-latency requirement of UHD VR video transmission \cite{DBLP:journals/tmm/HouDZB21,cheng2021design,DBLP:journals/twc/ChenSSLY20,DBLP:journals/tcom/SunCTL19}.
Therefore, this paper investigates the issue of maximizing users' feeling of presence (FoP) in their virtual world in a mmWave-enabled VR network incorporating CoMP transmission and machine learning.
The main contributions of this paper are summarized as follows:
\begin{itemize}
\item Owing to the user movement and the time-varying wireless channel conditions, we formulate the issue of maximizing users' FoP in virtual environments as a mixed-integer and sequential decision problem, subject to power consumption constraints on APs and users' HMDs. This problem is difficult to be directly solved by exploring conventional numerical optimization methods due to the lack of accurate users' tracking information (including users' locations and orientation angles) and mixed-integer and sequence-dependent characteristics.
\item As users' historical tracking information is separately collected by diverse APs, a parallel echo state network (ESN) learning method is designed to predict users' tracking information while accelerating the learning process.
\item With the predicted results, we develop a deep reinforcement learning (DRL) based optimization algorithm to tackle the mixed-integer and sequential decision problem. Particularly, to avoid generating infeasible solutions by simultaneously optimizing all variables while alleviating the curse of dimensionality issue, the DRL-based optimization algorithm decomposes the formulated mixed-integer optimization problem into an integer association optimization problem and a continuous power control problem.
    Next, deep neural networks (DNNs) with continuous action output spaces followed by an action quantization scheme are implemented to solve the integer association problem. Given the association results, the power control problem is solved to criticize them and optimize the transmit power.
\item Finally, the performance of the proposed DRL-based optimization algorithm is compared with various benchmark algorithms, and the impact of different design parameters is also discussed. Simulation results demonstrate the effectiveness of the proposed algorithm.
\end{itemize}

\section{System Model and problem formulation}
As shown in Fig. \ref{fig:fig_VR_Scenario}, we consider a mmWave-enabled VR network incorporating a CoMP transmission technique. This network includes a centralized unit (CU) connecting to $J$ distributed units (DUs) via optical fiber links, a set $\mathcal {J}$ of $J$ access points (APs) connected with the DUs, and a set of $\mathcal {U}$ of $N$ ground mobile users wearing HMDs.
To acquire immersive and interactive experiences, users will report their tracking information to their connected APs via reliable uplink communication links. Further, with collected users' tracking information, the CU will centrally simulate and construct virtual environments and coordinately transmit UHD VR videos to users via all APs in real time.
To accomplish the task of enhancing users' immersive and interactive experiences in virtual environments, joint uplink and downlink communications should be considered.
We assume that APs and users can work at both mmWave (exactly, 28 GHz) and sub-6 GHz frequency bands, where the mmWave frequency band is reserved for downlink UHD VR video delivery, and the sub-6 GHz frequency band is allocated for uplink users' tracking information transmission. This is because an ultra-high data rate can be achieved on the mmWave frequency band, and sub-6 GHz can support reliable communications.
Besides, to theoretically model the joint uplink and downlink communications, we suppose that the time domain is discretized into a sequence of time slots in the mmWave-enabled VR network and conduct the system modelling including uplink and downlink transmission models, FoP model, and power consumption model.
\begin{figure}[!t]
\begin{minipage}[t]{0.45\textwidth}
\centering
\includegraphics[width=2.6 in]{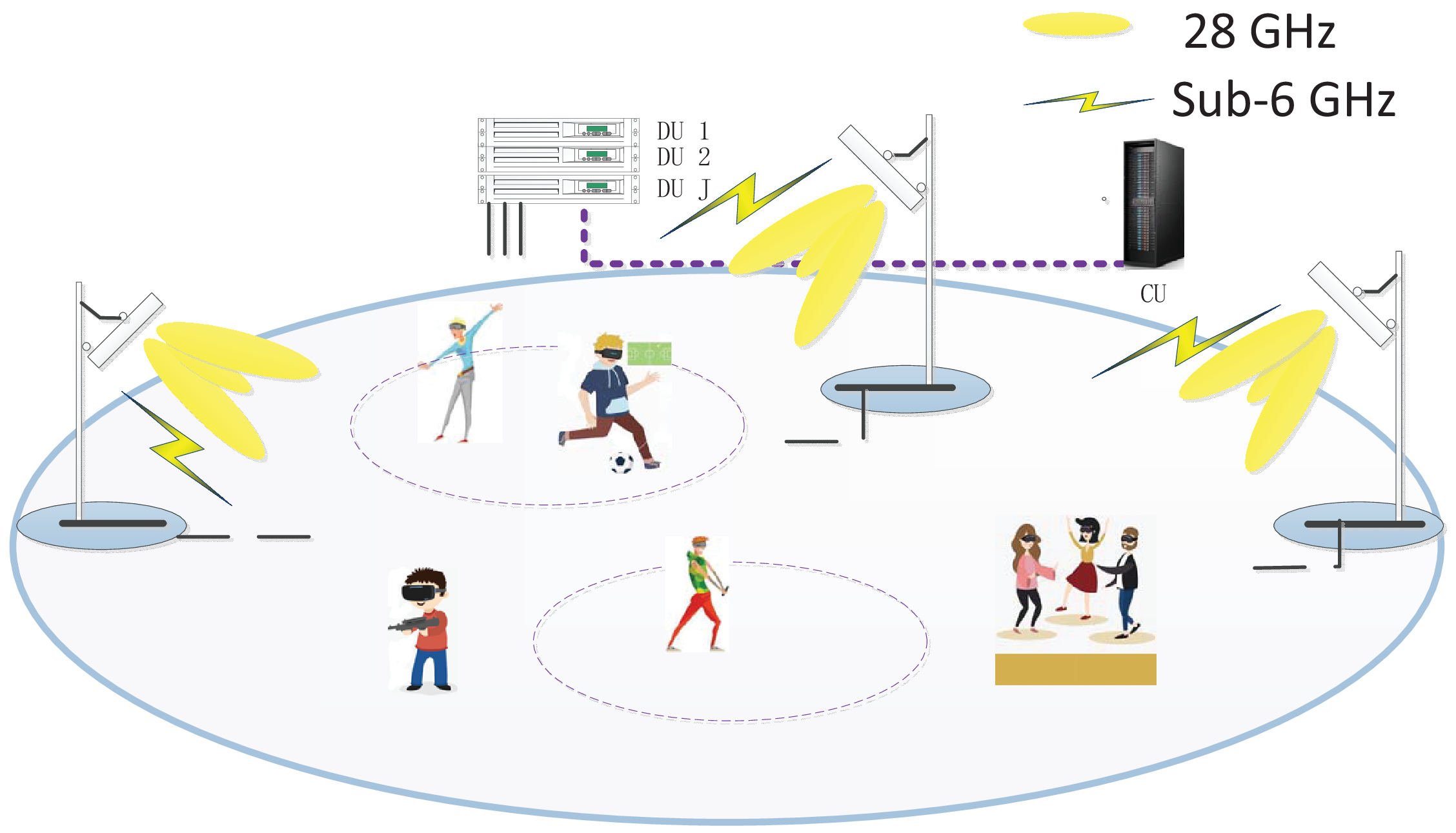}
\caption{A mmWave-enabled VR network incorporating CoMP transmission.}
\label{fig:fig_VR_Scenario}
\end{minipage}
\hspace{0.05\linewidth}
\begin{minipage}[t]{0.45\textwidth}
\centering
\includegraphics[width=2.6 in]{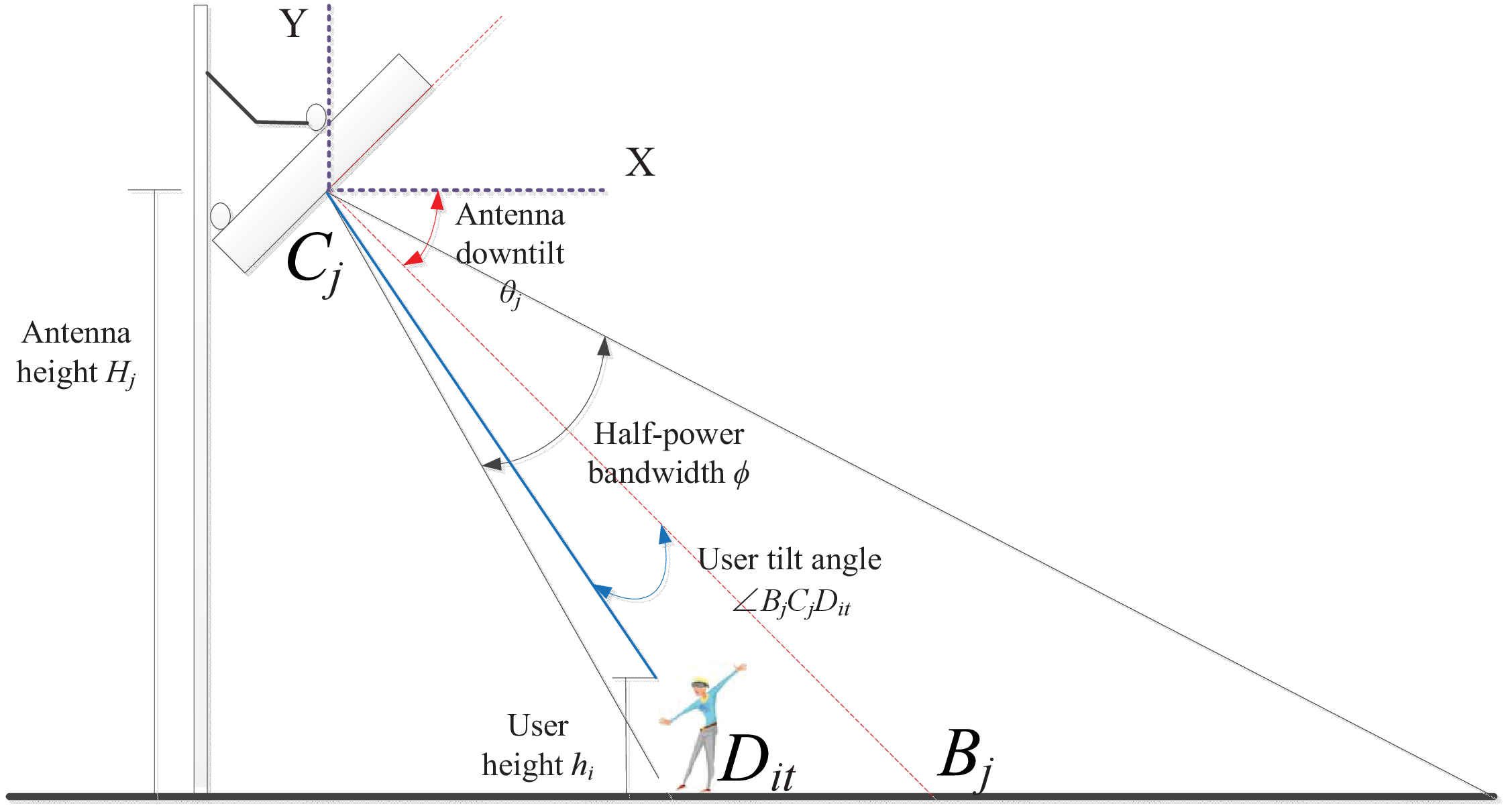}
\caption{Sectored antenna model of an AP.}
\label{fig:fig_VR_antenna_angle}
\end{minipage}
\end{figure}

\subsection{Uplink and downlink transmission models}
\subsubsection{Uplink transmission model}
Denote ${\bm x}_{it}^{\rm 3D}=[x_{it}, y_{it}, h_{i}]^{\rm T}$ as the three dimensional (3D) Cartesian coordinate of the HMD worn by user $i$ for all $i \in {\mathcal U}$ at time slot $t$ and $h_i \sim {\mathcal N}(\bar h, \sigma_h^2)$ is the user height.
$[x_{it}, y_{it}]^{\rm T}$ is the two dimensional (2D) location of user $i$ at time slot $t$.
Denote ${{\bm v}_{j}^{\rm 3D}} = [x_j, y_j, H_j]^{\rm T}$ as the 3D coordinate of the antenna of AP $j$ and $H_j$ is the antenna height.
Owing to the reliability requirement, users' data information (e.g., users' tracking information and profiles) is required to be successfully decoded by corresponding APs.
We express the condition that an AP can successfully decode the received user data packets as follows
\begin{equation}\label{eq:SINR_condition}
SNR_{ijt}^{\rm ul} = \frac{a_{ijt}^{\rm ul} p_{it}c_{ij}{\hat h_{ijt}}}{{N_0}W^{\rm ul}/N} \ge { \theta^{\rm th}}, \forall i, j, t,
\end{equation}
where $a_{ijt}^{\rm ul} \in \{0, 1\}$ is an association variable indicating whether user $i$'s uplink data packets can be successfully decoded by AP $j$ at time slot $t$. The data packets can be decoded if $a_{ijt}^{\rm ul} = 1$; otherwise, $a_{ijt}^{\rm ul} = 0$.
$p_{it}$ is the uplink transmit power of user $i$'s HMD, $c_{ij}$ is the Rayleigh channel gain,
$\hat h_{ijt} = {d_{ijt}^{ - \alpha }({{\bm x}_{it}^{\rm 3D}},{{\bm v}_{j}^{\rm 3D}})}$ is the uplink path-loss from user $i$ to AP $j$ with $\alpha$ being the fading exponent, $d_{ijt}(\cdot)$ denotes the Euclidean distance between user $i$ and AP $j$, $N_0$ denotes the single-side noise spectral density, $W^{\rm ul}$ represents the uplink bandwidth. $\theta^{\rm th}$ is the target signal-to-noise ratio (SNR) experienced at AP $j$ for successfully decoding data packets from user $i$. Besides, considering the reliability requirement of uplink transmission and the stringent power constraint on HMDs, frequency division multiplexing (FDM) technique is adopted in this paper. The adoption of FDM technique can avoid the decoding failure resulting from uplink signal interferences and significantly reduce power consumption without compensating the signal-to-interference-plus-noise ratio (SINR) loss caused by uplink interferences.

Additionally, we assume that each user $i$ can connect to at most one AP $j$ via the uplink channel at each time slot $t$, i.e., $\sum\nolimits_{j \in {\mathcal J}} {a_{ijt}^{\rm ul}}  \le 1$, $\forall i$. This is reasonable because it is unnecessary for each AP to decode all users' data successfully at each time slot $t$. A user merely connects to an AP (e.g., the nearest AP if possible) will greatly reduce power consumption.
Meanwhile, considering the stringent low-latency requirements of VR applications and the time consumption of processing (e.g., decoding and checking) received user data packets, we assume that an AP can serve up to $\tilde M$ users during a time slot, i.e., $\sum\nolimits_{i \in {\mathcal U}} {a_{ijt}^{\rm ul}}  \le \tilde M$, $\forall j$.

\subsubsection{Downlink transmission model}
In the downlink transmission configuration, antenna arrays are deployed to perform directional beamforming. For analysis facilitation, a sectored antenna model \cite{semiari2017inter-operator}, which consists of four components, i.e., the half-power beamwidth $\phi$, the antenna downtilt angle $\theta_j$ $\forall j$, the antenna gain of the mainlobe $G$, and the antenna gain of the sidelobe $g$, shown in Fig. \ref{fig:fig_VR_antenna_angle}, is exploited to approximate actual array beam patterns.
The antenna gain of the transmission link from AP $j$ to user $i$ is
\begin{equation}\label{eq:antenna_gain}
{f_{ijt}} = \left\{ {\begin{array}{*{20}{c}}
G&{ \angle B_jC_jD_{it} \le \frac{\phi }{2},}\\
g&{\rm otherwise,}
\end{array}} \right. \forall i, j, t,
\end{equation}
where $\angle B_jC_jD_{it}$ represents user $i$'s tilt angle towards AP $j$, the location of the point `$B_j$' can be determined by AP $j$'s 2D coordinate $\bm v_j^{\rm 2D} = [x_j, y_j]^{\rm T}$ and $\theta_j$, the point `$D_{it}$' represent user $i$'s position, the point `$C_j$' denotes the position of AP $j$'s antenna.

For any AP $j$, the 2D coordinate $\bm x_{bj}^{\rm 2D} = [x_{bj}, y_{bj}]^{\rm T}$ of point `$B_j$' can be given by
\begin{equation}\label{body_block}
{x_{bj}} = {d_j}({x_o} - {x_j})/{r_j} + {x_j}, \forall j,
\end{equation}
\begin{equation}\label{body_block}
{y_{bj}} = {d_j}({y_o} - {y_j})/{r_j} + {y_j}, \forall j,
\end{equation}
where $d_j = H_j/\tan(\theta_j)$, $r_j = ||\bm x_o - \bm v_j^{\rm 2D}||_2$, and $\bm x_o = [x_o, y_o]^{\rm T}$ is 2D coordinate of the center point of the considered communication area.

Then, user $i$'s tilt angle towards AP $j$ can be written as
\begin{equation}\label{eq:tilt_angle}
\angle B_jC_jD_{it} = \arccos \left( \frac{{\overrightarrow {{C_j}{B_j}}  \cdot \overrightarrow {{C_j}{D_{it}}} }}{{{{\left\| {{C_j}{B_j}} \right\|}_2}{{\left\| {{C_j}{D_{it}}} \right\|}_2}}} \right), \forall i, j, t,
\end{equation}
where direction vectors ${\overrightarrow {{C_j}{B_j}}} = (x_{bj}-x_j, y_{bj} - y_j, -H_j)$ and ${\overrightarrow {{C_j}{D_{it}}} } = (x_{it} - x_j, y_{it} - y_j, h_i - H_j)$.


A mmWave link may be blocked if a user turns around; this is because the user wears an HMD in front of his/her forehead. Denote $\vartheta $ as the maximum angle within which an AP can experience LoS transmission towards its downlink associated users.
For user $i$ at time slot $t$, an indicator variable $b_{ijt}$ introduced to indicate the blockage effect of user $i$'s body is given by
\begin{equation}\label{body_block}
{b_{ijt}} = \left\{ {\begin{array}{*{20}{c}}
1& \angle ({\vec A_{jit}},{\vec x_{it}}) > \vartheta,\\
0&{\rm otherwise,}
\end{array}} \right. \forall i,j,t,
\end{equation}
where $\angle ({\vec A_{jit}},{\vec x_{it}})$ represents the orientation angle of user $i$ at time slot $t$, which can be determined by locations of both user $i$ and AP $j$,\footnote{In this paper, we consider the case of determining users' orientation angles via the locations of both APs and users. Certainly, our proposed learning method is also applicable to scenarios where users' orientation angles need to be predicted.} ${\vec x}_{it} = (x_{it} - x_{it-1}, y_{it} - y_{it-1})$ is a direction vector. When $t =1$, the direction vector ${\vec x}_{i1} = (x_{i1}, y_{i1})$. ${\vec A_{jit}}=(x_{j} - x_{it}, y_{j} - y_{it})$ is a direction vector between the AP $j$ and user $i$.

Given ${\vec A_{jit}}$ and ${\vec x_{it}}$, we can calculate the orientation angle of user $i$ that is also the angle between ${\vec A_{jit}}$ and ${\vec x}_{it}$ by
\begin{equation}\label{eq:orientation_angle}
\angle ({\vec A_{jit}},{\vec x_{it}}) = \arccos \left( \frac{{\vec A_{jit}} \cdot {\vec x_{it}}}{||{\vec A_{jit}}||_2||{\vec x_{it}}||_2} \right), \forall i, j, t.
\end{equation}

The channel gain coefficient $h_{ijkt}$ of an LoS link and a non line-of-sight (NLoS) link between the $k$-th antenna element of AP $j$ and user $i$ at time slot $t$ can take the form \cite{semiari2017inter-operator}
\begin{equation}\label{eq:path_loss}
\begin{array}{l}
10{\log _{10}}({h_{ijkt}}h_{ijkt}^{\rm{H}}) = \left\{ {\begin{array}{*{20}{l}}
{\begin{array}{*{20}{l}}
{10{\eta _{{\rm{LoS}}}}{{\log }_{10}}({d_{ijt}}(x_{it}^{{\rm{3D}}},v_j^{{\rm{3D}}}))}
{ + 20{{\log }_{10}}\left( {\frac{{4\pi {f_c}}}{c}} \right) + } \\
\qquad \qquad {10{{\log }_{10}}{f_{ijt}} + \mu _k^{{\rm{LoS}}},}
\end{array}}&{{b_{ijt}} = 0}\\
{\begin{array}{*{20}{l}}
{10{\eta _{{\rm{NLoS}}}}{{\log }_{10}}({d_{ijt}}(x_{it}^{{\rm{3D}}},v_j^{{\rm{3D}}}))}
{ + 20{{\log }_{10}}\left( {\frac{{4\pi {f_c}}}{c}} \right) + } \\
\qquad \qquad {10{{\log }_{10}}{f_{ijt}} + \mu _k^{{\rm{NLoS}}},}
\end{array}}&{{b_{ijt}} = 1}
\end{array}} \right. \forall i,j,k,t,
\end{array}
\end{equation}
where $f_c$ (in Hz) is the carrier frequency, $c$ (in m/s) the light speed, $\eta_{\rm LoS}$ (in dB) and $\eta_{\rm NLoS}$ (in dB) the path-loss exponents of LoS and NLoS links, respectively, $\mu_{{\rm LoS}}  \sim {\mathcal {CN}}(0, \sigma_{\rm LoS}^2)$ (in dB) and $\mu_{{\rm NLoS}}  \sim {\mathcal {CN}}(0, \sigma_{\rm NLoS}^2)$ (in dB).

For any user $i$, to satisfy its immersive experience requirement, its downlink achievable data rate (denoted by ${r_{it}^{{\rm{dl}}}}$) from cooperative APs should be no less than a data rate threshold $\gamma^{\rm th}$, i.e.,
\begin{equation}\label{eq:data_rate_condition}
{r_{it}^{\rm dl} \ge { \gamma^{\rm th}} }, \text{ } \forall i, t.
\end{equation}

Define $a_{it}^{\rm dl} \in \{0, 1\}$ as an association variable indicating whether the user $i$'s data rate requirement can be satisfied at time slot $t$. $a_{it}^{\rm dl} = 1$ indicates that its data rate requirement can be satisfied; otherwise, $a_{it}^{\rm dl} = 0$. Then, for any user $i$ at time slot $t$, according to Shannon capacity formula and the principle of CoMP transmission, we can calculate ${r_{it}^{\rm dl}}$ by
\begin{equation}\label{eq:downlink_rate}
{r_{it}^{{\rm{dl}}} = {W^{{\rm{dl}}}}{{\log }_2}\left( {1 + \frac{{{a_{it}^{\rm dl}}{{| \sum\nolimits_{j\in {\mathcal J}}{{\bm h_{ijt}^{\rm H}}{\bm g_{ijt}}} |}^2}}}{{{N_0}{W^{{\rm{dl}}}} + I_{it}^{\rm dl}}}} \right)}, \text{ } \forall i, t,
\end{equation}
where $\bm h_{ijt} = [h_{ij1t}, \ldots, h_{ijKt}]^{\rm T} \in {\mathbb C}^{K}$ is a channel gain coefficient vector with $K$ denoting the number of antenna elements, $\bm g_{ijt} \in {\mathbb C}^{K}$ is the transmit beamformer pointed at user $i$ from AP $j$, $W^{\rm dl}$ represents the downlink system bandwidth. Owing to the directional propagation, for user $i$, not all users will be its interfering users. It is regarded that users whose distances from user $i$ are small than $D^{\rm th}$ will be user $i$'s interfering users, where $D^{\rm th}$ is determined by antenna configuration of APs (e.g., antenna height and downtilt angle). Denote the set of interfering users of user $i$ at time slot $t$ by ${\mathcal M}_{it}$, then, we have $I_{it}^{\rm dl}= \sum\nolimits_{m \in {\mathcal M}_{it}}{a_{mt}^{\rm dl} {{{| \sum\nolimits_{j\in {\mathcal J}}{{\bm h_{mjt}^{\rm H}}{\bm g_{mjt}}} |}^2}}}$.

\subsection{Feeling of presence model}
In VR applications, FoP represents an event that does not drag users back from engaging and immersive fictitious environments \cite{DBLP:journals/presence/BouchardSRR08}. For wireless VR, the degrading FoP can be caused by the collection of inaccurate users' tracking information via APs and the reception of low-quality VR video frames. Therefore, we consider the uplink user tracking information transmission and downlink VR video delivery when modelling the FoP experienced by users.
Mathematically, over a period of time slots, we model the FoP experienced by users as the following
\begin{equation}\label{eq:BIP_wireless_transmission}
\bar { B} (T) = \frac{1}{T}\sum\nolimits_{t = 1}^T {\left( {B_t^{\rm ul}\left( {\bm a_t^{\rm ul}} \right) + B_t^{\rm dl}\left( {\bm a_t^{\rm dl}} \right)} \right)},
\end{equation}
where ${{B}}_t^{{\rm{ul}}}\left( {\bm a_t^{{\rm{ul}}}} \right) = \frac{1}{{{N}}}\sum\limits_{i \in {\mathcal U}} {\sum\limits_{j \in {\mathcal J}} {a_{ijt}^{{\rm{ul}}}} } $ with $\bm a_{t}^{\rm ul} = [a_{11t}^{\rm ul}, \ldots, a_{ijt}^{\rm ul}, \ldots, a_{NJt}^{\rm ul}]^{\rm T}$, $B_t^{{\rm{dl}}}\left( {\bm a_t^{{\rm{dl}}}} \right) = \frac{1}{N} {\sum\limits_{i \in {{\mathcal U}}} {a_{it}^{{\rm{dl}}}} } $ with $\bm a_{t}^{\rm dl} = [a_{1t}^{\rm dl}, \ldots, a_{it}^{\rm dl}, \ldots, a_{Nt}^{\rm dl}]^{\rm T}$.

\subsection{Power consumption model}
HMDs are generally battery-driven and constrained by the maximum instantaneous power. For any user $i$'s HMD, define $p_{it}^{\rm tot}$ as its instantaneous power consumption including the transmit power and circuit power consumption (e.g., power consumption of mixers, frequency synthesizers, and digital-to-analog converters) at time slot $t$, we then have
\begin{equation}\label{eq:energy_consumption}
{p_{it}^{\rm tot}}  \le {\tilde p_i}, \forall i, t,
\end{equation}
where $p_{it}^{\rm tot} = {p_{it}} + p_i^c$, $p_i^c$ denotes the HMD's circuit power consumption during a time slot, and $\tilde p_i$ is a constant.
Without loss of generality, we assume that all users' HMDs are homogenous.

The instantaneous power consumption of each AP is also constrained. As CoMP transmission technique is explored, for any AP $j$, we can model its instantaneous power consumption at time slot $t$ as the following
\begin{equation}\label{eq:power_consumption}
{\sum\nolimits_{i \in {\mathcal U}}{a_{it}^{\rm dl} \bm g_{ijt}^{\rm H} \bm g_{ijt}}} +  E_j^c \le \tilde E_j, \forall j, t,
\end{equation}
where $ E_j^c$ is a constant representing the circuit power consumption, $\tilde E_j$ is the maximum instantaneous power of AP $j$.

\subsection{Objective function and problem formulation}
To guarantee immersive and interactive VR experiences of users over a period of time slots, uplink user data packets should be successfully decoded, and downlink data rate requirements of users should be satisfied at each time slot; that is, users' FoP should be maximized.
According to (\ref{eq:SINR_condition}) and (\ref{eq:BIP_wireless_transmission}), one might believe that increasing the transmit power of users' HMDs would be an appropriate way of enhancing users' FoP.
However, as users' HMDs are usually powered by batteries, they are encouraged to work in an energy-efficient mode to prolong their working duration. Further, reducing HMDs' power consumption indicates less heat generation, which can enhance users' VR experiences.
Therefore, our goal is to maximize users' FoP while minimizing the power consumption of HMDs over a period of time slots.
Combining with the above analysis, we can formulate the problem of enhancing users' immersive experiences as below
\begin{subequations}\label{eq:original_problem_p0}
\begin{alignat}{2}
& \mathop {\rm maximize}\limits_{\{\bm a_{t}^{\rm ul}, \bm a_{t}^{\rm dl}, \bm p_t, \bm g_{ijt}\}} \text{ } \mathop {\lim \inf }\limits_{T \to \infty } \frac{1}{T}\sum\limits_{t = 1}^T {\left(B_t^{{\rm{ul}}}\left( {\bm a_t^{{\rm{ul}}}} \right) + B_t^{{\rm{dl}}}\left( {\bm a_t^{{\rm{dl}}}} \right)\right)}  -  \frac{1}{T}\sum\limits_{t = 1}^T {\sum\limits_{i \in {\mathcal U}} {\sum\limits_{j \in {\cal J}} {a_{ijt}^{{\rm{ul}}}p_{it}^{{\rm{tot}}}/{{\tilde p}_i}} } } \allowdisplaybreaks[4]\\
& {\rm s.t.} \quad \sum\nolimits_{j \in {\mathcal J}} {a_{ijt}^{\rm ul}}  \le 1, \forall i, t \allowdisplaybreaks[4] \\
& \qquad \sum\nolimits_{i \in {\mathcal U}} {a_{ijt}^{\rm ul}}  \le \tilde M, \forall j, t \\
& \qquad {a_{ijt}^{\rm ul}} \in \{0, 1\}, \forall i, j, t \\
& \qquad {a_{it}^{\rm dl}} \in \{0, 1\}, \forall i, t \\
& \qquad 0 \le p_{it} \le \tilde p_{i} - p_i^c, \forall i, t \\
& \qquad \rm {(\ref{eq:SINR_condition}), (\ref{eq:data_rate_condition}), (\ref{eq:power_consumption}),}
\end{alignat}
\end{subequations}
where $\bm p_t=[p_{1t},p_{2t},\ldots,p_{Nt}]^{\rm T}$.

However, the solution to (\ref{eq:original_problem_p0}) is highly challenging due to the unknown users' tracking information at each time slot. Given users' tracking information, the solution to (\ref{eq:original_problem_p0}) is still NP-hard or even non-detectable. 
It can be confirmed that (\ref{eq:original_problem_p0}) is a mixed-integer non-linear programming (MINLP) problem as it simultaneously contains zero-one variables, continuous variables, and non-linear constraints. Further, we can know that (\ref{eq:data_rate_condition}) and (\ref{eq:power_consumption}) are non-convex with respect to (w.r.t) $a_{it}^{\rm dl}$ and ${\bm g}_{ijt}$, $\forall i$, $j$, by evaluating the Hessian matrix.
To tackle the tricky problem, we develop a novel solution framework as depicted in Fig. \ref{fig:fig_solution_framework}. In this framework, we first propose to predict users' tracking information using a machine learning method.
With the predicted results, we then develop a DRL-based optimization algorithm to solve the MINLP problem.
The procedure of solving (\ref{eq:original_problem_p0}) is elaborated in the following sections.
\begin{figure}[!t]
\begin{minipage}[t]{0.45\textwidth}
\centering
\includegraphics[width=2.6 in]{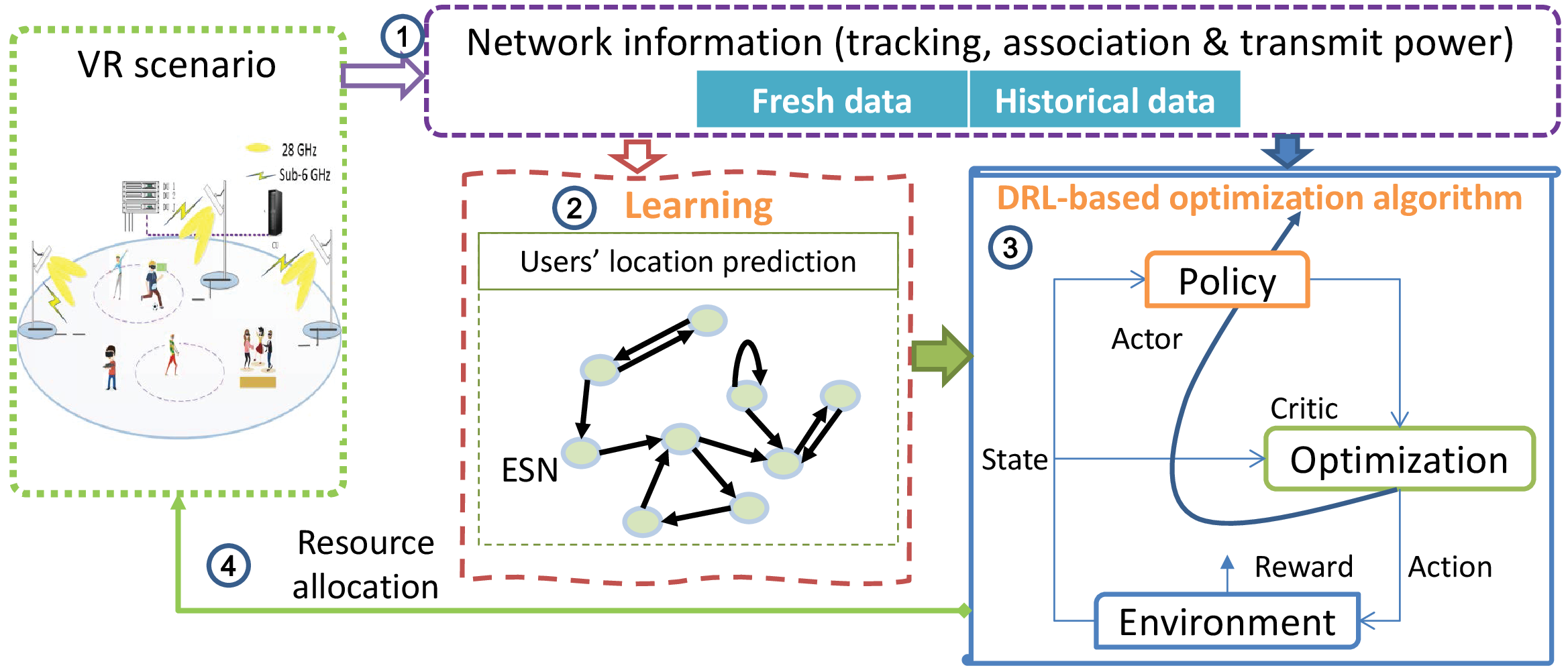}
\caption{Working diagram of a framework of solving (\ref{eq:original_problem_p0}).}
\label{fig:fig_solution_framework}
\end{minipage}
\hspace{0.05\linewidth}
\begin{minipage}[t]{0.45\textwidth}
\centering
\includegraphics[width=2.4 in]{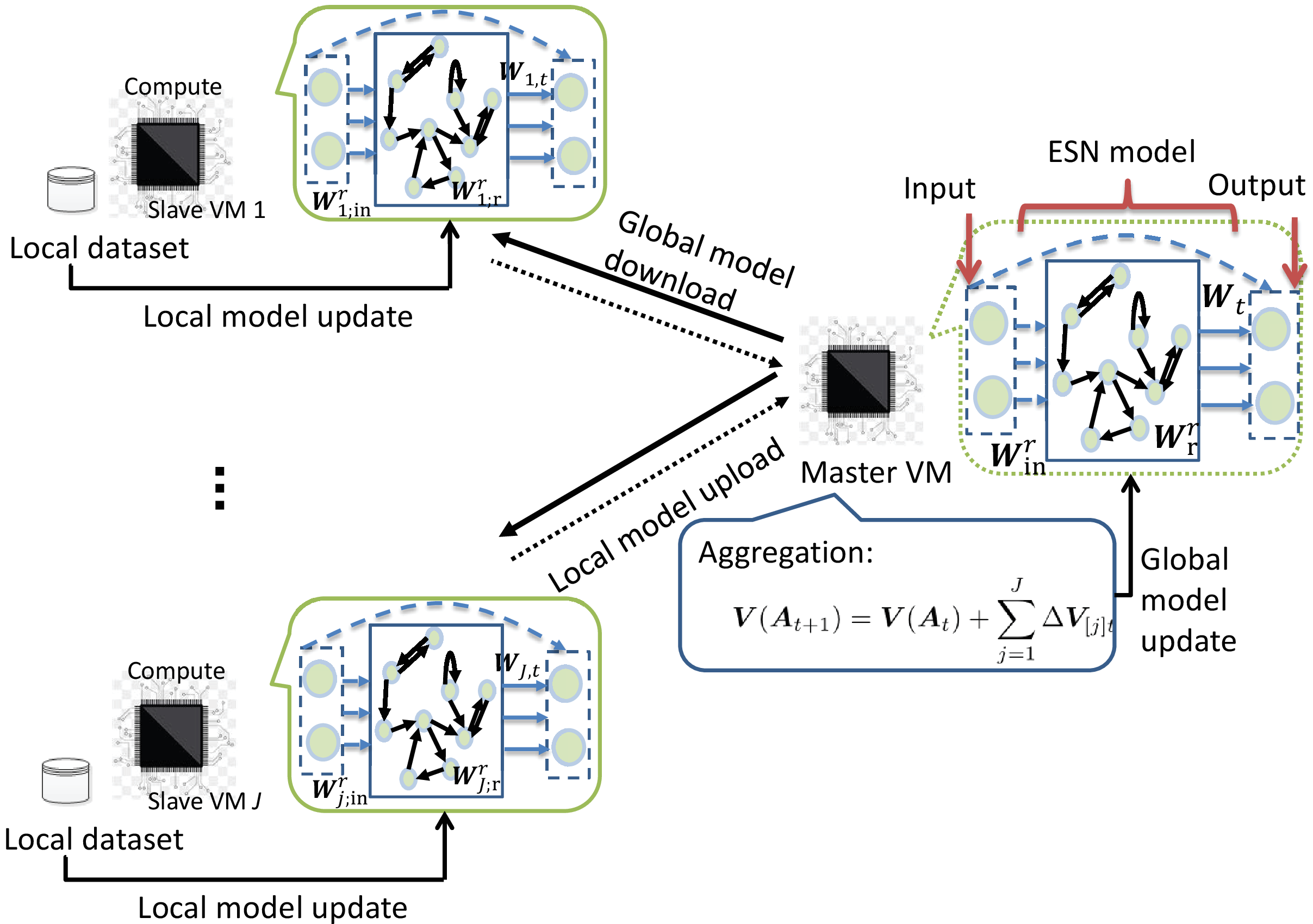}
\caption{Architecture of the parallel ESN learning method.}
\label{fig:fig_Federated_ESN_architecture}
\end{minipage}
\end{figure}

\section{Users' Location Prediction}
As analyzed above, the efficient user-AP association and transmit power of both HMDs and APs are configured on the basis of the accurate perception of users' tracking information. If the association and transmit power are identified without knowledge of users' tracking information, users may have degrading VR experiences, and the working duration of users' HMDs may be dramatically shortened.
Meanwhile, owing to the stringent low latency requirement, the user-AP association and transmit power should be proactively determined to enhance users' immersive and interactive VR experiences.
Hence, APs must collect fresh and historical tracking information for users' tracking information prediction in future time slots.
With predicted tracking information, the user-AP association and transmit power can be configured in advance.
Certainly, from (\ref{eq:orientation_angle}), we observe that users' orientation angles can be obtained by their and APs' locations; thus, we only predict users' locations in this section.
Machine learning is convinced as a promising proposal to predict users' locations. In machine learning methods, the accuracy and completeness of sample collection are crucial for accurate model training.
However, the user-AP association may vary with user movement, which indicates that location information of each user may scatter in multiple APs, and each AP may only collect partial location information of its associated users after a period of time.
To tackle this issue, we develop a parallel machine learning method, which exploits $J$ slave virtual machines (VMs) created in the CU to train learning models for each user, as shown in Fig. \ref{fig:fig_Federated_ESN_architecture}. Besides, for each AP, it will feed its locally collected location information to a slave VM for training.
In this way, the prediction process can also be accelerated.
With the predicted results, the CU can then proactively allocate system resources by solving (\ref{eq:original_problem_p0}).

\subsection{Echo state network}
In this section, the principle of echo state network (ESN) is exploited to train users' location prediction model as the ESN method can efficiently analyze the correlation of users' location information and quickly converge to obtain users' predicted locations \cite{DBLP:journals/nn/ScardapaneWP16}.
It is noteworthy that there are some differences between the traditional ESN method and the developed parallel ESN learning method.
The traditional ESN method is a centralized learning method with the requirement of the aggregation of all users' locations scattered in all APs, which is not required for the parallel ESN learning method.
What's more, the traditional ESN method can only be used to conduct data prediction in a time slot while the parallel ESN learning method can predict users' locations in $M > 1$ time slots.
An ESN is a recurrent neural network that can be partitioned into three components: input, ESN model, and output, as shown in Fig. \ref{fig:fig_Federated_ESN_architecture}.
For any user $i \in {\mathcal U}$, the $N_i$-dimensional input vector $\bm x_{it} \in {\mathbb R}^{N_i }$ is fed to an $N_r$-dimensional reservoir whose internal state $\bm s_{i(t-1)} \in {\mathbb R}^{N_r }$ is updated according to the state equation
\begin{equation}\label{eq:state}
\bm s_{it} = \tanh\left( \bm W_{\rm in}^r \bm x_{it} + \bm W_{\rm r}^r \bm s_{i(t-1)} \right),
\end{equation}
where $\bm W_{\rm in}^r \in {\mathbb R}^{N_r \times N_i}$ and $\bm W_{\rm r}^r \in {\mathbb R}^{N_r \times N_r}$ are randomly generated matrices with each matrix element locating in the interval $(0, 1)$.

The evaluated output of the ESN at time slot $t$ is given by
\begin{equation}\label{eq:output}
\hat {\bm y}_{i(t+1)} = \bm W_{\rm in}^o \bm x_{it} + \bm W_{\rm r}^o \bm s_{it},
\end{equation}
where $\bm W_{\rm in}^o \in {\mathbb R}^{N_o \times N_i}$, $\bm W_{\rm r}^o \in {\mathbb R}^{N_o \times N_r}$ are trained based on collected training data samples. 

To train the ESN model, suppose we are provided with a sequence of $Q$ desired input-outputs pairs $\{(\bm x_{i1}, {\bm y}_{i1}), \ldots, (\bm x_{iQ}, {\bm y}_{iQ})\}$ of user $i$, where $\bm y_{it} \in {\mathbb R}^{N_o }$ is the target location of user $i$ at time slot $t$.
Define the hidden matrix $\bm X_{it}$ as
\begin{equation}\label{eq:hidden_matrix}
{\bm X_{it}} = \left[ {\begin{array}{*{20}{c}}
\begin{array}{l}
{\bm x_{i1}}\\
{\bm s_{i1}}
\end{array}& \cdots &\begin{array}{l}
{\bm x_{iQ}}\\
{\bm s_{iQ}}
\end{array}
\end{array}} \right].
\end{equation}

The optimal output weight matrix is then achieved by solving the following regularized least-square problem
\begin{equation}\label{eq:least_square_prob}
\bm W_t^{\star} = \mathop {\arg \min }\limits_{\bm W_t \in {{\mathbb R}^{({N_i} + {N_r}) \times {N_o}}}} \text{ }  \frac{1}{Q}l\left(\bm X_{it}^{\rm T} \bm W_t \right ) + \xi r (\bm W_t)
\end{equation}
where $\bm W_t = [\bm W_{\rm in}^o \bm W_{\rm r}^o]^{\rm T}$, $\xi \in {\mathbb R}_{+}$ is a positive scalar known as regularization factor, the loss function $l(\bm X_{it}^{\rm T} \bm W_t) = \frac{1}{2} || {\bm X_{it}^{\rm T} \bm W_t - \bm Y_{it}} ||_F^2$, the regulator $r(\bm W_t) = || \bm W_t ||_F^2$, and the target location matrix $\bm Y_{it} = [\bm y_{i1}^{\rm T}; \ldots; \bm y_{iQ}^{\rm T}] \in {\mathbb R}^{Q \times N_o}$.

\subsection{Parallel ESN learning method for users' location prediction}
Based on the principle of the ESN method, we next elaborate on the procedure of the parallel ESN learning method for users' location prediction.
To facilitate the analysis, we make the following assumptions on the regulator and the loss function.
\begin{assumption}\label{as:assumption_r_function}
\rm {
The function $r: {\mathbb R}^{m \times n} \to {\mathbb R}$ is $\zeta $-strongly convex, i.e., $\forall i \in \{1, 2, \ldots, n\}$, $\forall \bm X $, and $\Delta \bm X \in {\mathbb R}^{m \times n}$, we have \cite{DBLP:journals/tcom/YangLQP20}
\begin{equation}\label{eq:r_function_expansion}
r(\bm X + \Delta \bm X) \ge r(\bm X) + \nabla r(\bm X) \odot \Delta \bm X + {\zeta} ||\Delta \bm X||_F^2/{2},
\end{equation}
where $\nabla r(\cdot)$ denotes the gradient of $r(\cdot)$.
}
\end{assumption}

\begin{assumption}\label{as:assumption_loss_function}
\rm {
The function $l : {\mathbb R} \to {\mathbb R}$ are $\frac{1}{\mu}$-smooth, i.e., $\forall i \in \{1, 2, \ldots, n\}$, $\forall x$, and $\Delta x \in {\mathbb R}$, we have
\begin{equation}\label{eq:loss_function_expansion}
l(x + \Delta x) \le l(x) + \nabla l (x) \Delta x + (\Delta x)^2/{2\mu},
\end{equation}
where $\nabla l(\cdot)$ represents the gradient of $l(\cdot)$.
}
\end{assumption}

According to Fenchel-Rockafeller duality, we can formulate the local dual optimization problem of (\ref{eq:least_square_prob}) in the following way.

\begin{lemma}\label{lem:lemma_1}
\rm {For a set of $J$ slave VMs and a typical user $i$, the dual problem of (\ref{eq:least_square_prob}) can be written as follows
\begin{equation}\label{eq:dual_problem}
\mathop {\rm maximize }\limits_{\bm A \in {{\mathbb R}^{Q \times {N_o}}}} \text{ }  \left\{ { - \xi r^{\star}\left( {\frac{1}{{\xi Q}}{\bm A^{\rm T}}{\bm X^{\rm T}}} \right) - \frac{1}{Q}\sum\limits_{m = 1}^{Q} {\sum\limits_{n = 1}^{{N_o}} {l^{\star}( - {a_{mn}})} } } \right\}
\end{equation}
where
\begin{equation}\label{eq:r_star}
{r^{\star}}(\bm C) = \frac{1}{4} {\sum\nolimits_{n=1}^{N_o}{\bm z_n^{\rm T} \bm C \bm C^{\rm T} \bm z_n} },
\end{equation}
\begin{equation}\label{eq:l_star}
l^{\star}(-a_{mn}) = {-a_{mn}y_{mn} + a_{mn}^2/{2} },
\end{equation}
$\bm A \in {\mathbb R}^{Q\times N_o}$ is a Lagrangian multiplier matrix, $\bm z_n \in {\mathbb R}^{N_o }$ is a column vector with the $n$-th element being one and all other elements being zero, $\bm X$ is a lightened notation of ${\bm X_{it}} =$ $\left[ {\begin{array}{*{20}{c}}
\begin{array}{l}
{\bm x_{i(t - 1)}}\\
{\bm s_{i(t - 1)}}
\end{array}& \cdots &\begin{array}{l}
{\bm x_{i(t - Q)}}\\
{\bm s_{i(t - Q)}}
\end{array}
\end{array}} \right]$,
and $y_{mn}$ is an element of matrix ${\bm Y} = [\bm y_{it}^{\rm T}; \ldots; \bm y_{i(t-Q+1)}^{\rm T}]$ at the location of the $m$-th row and the $n$-th column.
}
\end{lemma}

\begin{proof}
Please refer to Appendix A.
\end{proof}

Denote the objective function of (\ref{eq:dual_problem}) as $D(\bm A)$, and define $\bm V(\bm A) := \frac{1}{\xi Q} {(\bm X \bm A)^{\rm T}} \in {\mathbb R}^{N_o \times (N_i + N_r)}$, we can then rewrite $D(\bm A)$ as
\begin{equation}\label{eq:DA_decomposition}
D(\bm A) = -\xi r^{\star}(\bm V(\bm A)) - \sum\nolimits_{j = 1}^{J} {R_j( \bm A_{[j]})},
\end{equation}
where $R_j(\bm A_{[j]}) = \frac{1}{Q}\sum\limits_{m \in {\mathcal Q}_j} {\sum\limits_{{n} = 1}^{N_o} {l^{\star}( - {a_{mn}})} } $, $\bm A_{[j]}= \hat {\bm Z}_j \bm A$, and $\hat {\bm Z}_{j} \in {\mathbb R}^{Q \times Q}$ is a square matrix with $J \times J$ blocks. In $\hat {\bm Z}_j$, the block in the $j$-th row and $j$-th column is a $q_j \times q_j$ identity matrix with $q_j$ being the cardinality of a set ${\mathcal Q}_j$ and all other blocks are zero matrices, ${\mathcal Q}_j$ is an index set including the indices of $Q$ data samples fed to slave VM $j$.

Then, for a given matrix $ {\bm A}^t$, varying its value by $\Delta {\bm A}^t$ will change (\ref{eq:DA_decomposition}) as below
\begin{equation}\label{eq:Delta_A}
\begin{array}{l}
D( {\bm A}^t + \Delta  {\bm A}^t) = -\xi r^{\star}(\bm V( {\bm A}^t + \Delta  {\bm A}^t)) - \sum\nolimits_{j = 1}^{J} {R_j(  {\bm A}_{[j]}^t + \Delta {\bm A}_{[j]}^t)},
\end{array}
\end{equation}
where $\Delta {\bm A}_{[j]}^t = \hat {\bm Z}_j \Delta {\bm A}^t$.

Note that the second term of the right-hand side (RHS) of (\ref{eq:Delta_A}) includes the local changes of each VM $j$, while the first term involves the global variations.

As $r(\cdot)$ is $\zeta$-strongly convex, $r^{\star}(\cdot)$ is then $\frac{1}{\zeta}$-smooth \cite{DBLP:journals/tcom/YangLQP20}. Thus, we can calculate the upper bound of $r^{\star}(V( {\bm A}^t + \Delta  {\bm A}^t))$ as follows
\begin{equation}\label{eq:r_bound}
\begin{array}{l}
r^{\star}(V({\bm A}^t + \Delta {\bm A}^t)) \le r^{\star}\left( {\bm V({\bm A}^t)} \right) + \frac{1}{\xi Q}\sum\limits_{n = 1}^{{N_o}} \bm z_n^{\rm T} {\nabla r^{\star}(\bm V({\bm A}^t)) \bm X \Delta {\bm A}^t{\bm z_n}} +  \frac{\kappa }{{2{{\left( {\xi Q} \right)}^2}}}\sum\limits_{n = 1}^{{N_o}} {{{\left\| {\bm X \Delta {\bm A}^t{\bm z_n}} \right\|}^2}} \\
 = r^{\star}\left( {\bm V({\bm A}^t)} \right) + \frac{1}{\xi Q}\sum\limits_{j = 1}^J {\sum\limits_{n = 1}^{{N_o}} \bm z_n^{\rm T} {\nabla r^{\star}(\bm V({\bm A}^t)) \bm X_{[j]} \Delta {\bm A}_{[j]}^t {\bm z_n}} }
 + \frac{\kappa }{{2{{\left( {\xi Q} \right)}^2}}}\sum\limits_{j = 1}^J {\sum\limits_{n = 1}^{{N_o}} {{{\left\| {\bm X_{[j]} \Delta {\bm A}_{[j]}^t {\bm z_n}} \right\|}^2}} },
\end{array}
\end{equation}
where $\bm X_{[j]} = \bm X \hat {\bm Z}_j$, $\kappa >\frac{1}{\zeta}$ is a data dependent constant measuring the difficulty of the partition to the whole samples.

By substituting (\ref{eq:r_bound}) into (\ref{eq:Delta_A}), we obtain
\begin{equation}\label{eq:D_bound}
\begin{array}{l}
D({\bm A}^t + \Delta {\bm A}^t) \ge - {\xi} r^{\star}\left( {\bm V({\bm A}^t)} \right) - \frac{1}{Q} \sum\limits_{j = 1}^J {\sum\limits_{n = 1}^{{N_o}} \bm z_n^{\rm T}{\nabla r^{\star}(\bm V({\bm A}^t)) \bm X_{[j]} \Delta {\bm A}_{[j]}^t {\bm z_n}} } \\
\qquad \qquad \qquad \quad -\frac{\kappa }{{2 \xi {{{ Q} }^2}}}\sum\limits_{j = 1}^J {\sum\limits_{n = 1}^{{N_o}} {{{\left\| { \bm X_{[j]} \Delta {\bm A}_{[j]}^t {\bm z_n}} \right\|}^2}} } - \sum\limits_{j = 1}^{J} {R_j( {\bm A}_{[j]}^t + \Delta {\bm A}_{[j]}^t)}.
\end{array}
\end{equation}

From (\ref{eq:D_bound}), we observe that the problem of maximizing $D({\bm A}^t + \Delta {\bm A}^t)$ can be decomposed into $J$ subproblems, and $J$ slave VMs can then be exploited to optimize these subproblems separately.
If slave VM $j$ can optimize $\Delta {\bm A}^t$ using its collected data samples by maximizing the RHS of (\ref{eq:D_bound}), the resultant improvements can be aggregated to drive $\bm D({\bm A}^t)$ toward the optimum. The detailed procedure is described below.

As shown in Fig. \ref{fig:fig_Federated_ESN_architecture}, during any communication round $t$, a master VM produces $\bm V(\bm A^t)$ using updates received at the last round and shares it with all slave VMs. The task at any slave VM $j$ is to obtain $\Delta \bm A_{[j]}^t$ by maximizing the following problem
\begin{equation}\label{eq:original_prob}
\begin{array}{l}
\Delta {\bm A_{[j]}^{t\star}} = \mathop {\arg \max }\limits_{\Delta {\bm A_{[j]}^t} \in {{\mathbb R}^{Q \times {N_o}}}} \text{ } \Delta {D_j}\left( {\Delta {\bm A_{[j]}^t};\bm V({\bm A^t}),{\bm A_{[j]}^t}} \right)\\
 = \mathop {\arg \max }\limits_{\Delta {\bm A_{[j]}^t} \in {{\mathbb R}^{Q \times {N_o}}}} \left\{ { - {R_j}\left( {{\bm A_{[j]}^t} + \Delta {\bm A_{[j]}^t}} \right) - \frac{\xi }{J}r^{\star}(\bm V({\bm A^t}))} \right.
\left.  - \frac{1}{Q} \sum\limits_{n = 1}^{{N_o}} \bm z_n^{\rm T} {\nabla r^{\star}(\bm V({\bm A^t})) \bm X_{[j]} \Delta {\bm A_{[j]}^t}{\bm z_n}}  - \right. \\
\quad \left. \frac{\kappa }{{2\xi {Q^2}}}\sum\limits_{n = 1}^{{N_o}} {{{\left\| { \bm X_{[j]}\Delta {\bm A_{[j]}^t}{\bm z_n}} \right\|}^2}}  \right\}.
\end{array}
\end{equation}

Calculate the derivative of $\Delta {D_j}\left( {\Delta {\bm A_{[j]}^t};\bm V({\bm A^t}),{\bm A_{[j]}^t}} \right)$ over $\Delta {\bm A_{[j]}^t}$, and force the derivative result to be zero, we have
\begin{equation}\label{eq:delta_a_expression}
\begin{array}{l}
\Delta {\bm A}_{[j]}^{t\star} = {\left( {{{ {{\hat {\bm Z}_j}} }} + \frac{\kappa }{{\xi Q}} \bm X_{[j]}^{\rm T}{\bm X_{[j]}}} \right)^{ - 1}} \left( {{\bm Y_{[j]}} - \bm A_{[j]}^t - \frac{1}{2}\bm X_{[j]}^{\rm T}{\bm V^{\rm T}}({\bm A_t})} \right),
\end{array}
\end{equation}
where $\bm Y_{[j]} = \hat {\bm Z}_j \bm Y$.

Next, slave VM $j$, $\forall j$, sends $\Delta \bm V_{[j]}^t = \frac{1}{{\xi Q}} {{{( {\bm X_{[j]}\Delta {\bm A_{[j]}^{t\star}}} )}^{\rm T}}} $ to the master VM. The master VM updates the global model as $\bm V(\bm A^t + \Delta \bm A^t) =  \bm V({\bm A^t}) + \sum\nolimits_{j = 1}^J {\Delta \bm V_{[j]}^t} $. Finally, alteratively update $\bm V(\bm A ^ t)$ and $\{\Delta \bm A_{[j]}^{t\star}\}_{j=1}^J$ on the global and local sides, respectively. It is expected that the solution to the dual problem can be enhanced at every step and will converge after several iterations.

At time slot $t$, based on the above derivation, the parallel ESN learning method for predicting locations of user $i$, $\forall i$, in $M$ time slots can be summarized in Algorithm \ref{Alg:Alg1}.
\begin{algorithm}
\caption{Parallel ESN learning for user location prediction}
\label{Alg:Alg1}
\begin{algorithmic}[1]
\STATE \textbf{Initialization:} Data samples of all slave VMs. For any slave VM $j$, it randomly initiates a starting point $ \bm A_{[j]}^0 \in {\mathbb R}^{Q \times N_o}$. The master VM collects $\frac{1}{\xi Q} {( \bm X_{[j]} \bm A_{[j]}^0)^{\rm T}}$ from all slave VMs, generates $\bm V(\bm A^0)= \sum_{j=1}^J {\frac{1}{\xi Q} {( \bm X_{[j]} \bm A_{[j]}^0)^{\rm T}}}$, and then share the model $\bm V(\bm A^0)$ with all slave VMs. Let $\kappa = J/{\zeta}$.
\FOR{$r = 0 : \bar r_{\rm max}-1$}
    \FOR{each slave VM $j \in \{1, 2, \ldots, J\}$ in parallel}
    \STATE {Calculate $\Delta {\bm A}_{[j]}^{r\star}$ using (\ref{eq:delta_a_expression})}, update and store the local Lagrangian multiplier
    \begin{equation}\label{eq:lagrangian_update}
    \bm A_{[j]}^{r+1} = \bm A_{[j]}^r +  \Delta \bm A_{[j]}^{r\star}/({r+1}).
    \end{equation}
    \STATE Compute the following local model and send it to the master VM
    \begin{equation}\label{eq:V_update}
    \Delta \bm V_{[j]}^r =  {{\left( { \bm X_{[j]} \Delta {\bm A_{[j]}^{r\star}}} \right)^{\rm T}}}/{{\xi Q}}.
    \end{equation}
    \ENDFOR
\STATE Given local models, the master VM updates the global model as
\begin{equation}\label{eq:v_t_update}
\bm V(\bm A^{r + 1}) = \bm V(\bm A^r) + \sum\nolimits_{j = 1}^J {\Delta \bm { V}_{[j]}^r},
\end{equation}
and then share the updated global model $\bm V(\bm A^{r+1})$ with all slave VMs.
\ENDFOR
\STATE Let $\bm W^{\rm T} = {\nabla r^{\star}(\bm V({\bm A^r}))}$ and predict user $i$'s location $\hat {\bm y}_{it}$ by (\ref{eq:output}). Meanwhile, by iteratively assigning $\bm x_{i(t+1)} = \hat {\bm y}_{it}$, each user $i$'s locations in $M$ time slots can be obtained.
\STATE \textbf{Output:} The predicted locations of user $i$, $\hat {\bm Y}_{it}= [\hat {\bm y}_{i(t+1)}^{\rm T}; \ldots; \hat {\bm y}_{i(t+M)}^{\rm T}]$, $\forall i$.
\end{algorithmic}
\end{algorithm}

\section{DRL-based Optimization Algorithm}
Given the predicted locations of all users, it is still challenging to solve the original problem owing to its non-linear and mixed-integer characteristics.
Alternative optimization is extensively considered as an effective scheme of solving MINLP problems.
Unfortunately, the popular alternative optimization scheme cannot be adopted in this paper. This is because the alternative optimization scheme is of often high computational complexity, and the original problem is also a sequential decision problem requiring an MINLP problem to be solved at each time slot. Remarkably, calling an optimization scheme with a high computational complexity at each time slot is unacceptable for latency-sensitive VR applications.

Reinforcement learning methods can be explored to solve sequential decision problems.
For example, the works in \cite{DBLP:journals/twc/BennisPBHP13,DBLP:journals/tvt/YangCXDXW19} proposed reinforcement learning methods to solve sequential decision problems with a discrete decision space and a continuous decision space, respectively.
However, how to solve sequential decision problems simultaneously involving discrete and continuous decision variables (e.g., the problem (\ref{eq:original_problem_p0})) is a significant and understudied problem.

In this paper, we propose a deep reinforcement learning (DRL)-based optimization algorithm to solve (\ref{eq:original_problem_p0}).
Specifically, we design a DNN joint with an action quantization scheme to produce a set of association actions of high diversity. Given the association actions, a continuous optimization problem is solved to criticize them and optimize the continuous variables. The detailed procedure is presented in the following subsections.

\subsection{Vertical decomposition}
Define a vector $\bm g_{it} = [\bm g_{i1t}; \ldots; \bm g_{ijt}; \ldots; \bm g_{iJt}] \in {\mathbb C}^{JK}$ and a vector $\bm h_{it} = [f_{i1t}\bm h_{i1t}; \ldots; f_{ijt}\bm h_{ijt};$ $\ldots; f_{iJt}\bm h_{iJt}] \in {\mathbb C}^{JK}$, $\forall i$, $t$. Let matrix $\bm G_{it} = \bm g_{it} \bm g_{it}^{\rm T}$ and matrix $\bm H_{it} = \bm h_{it} \bm h_{it}^{\rm T} $.
As ${\rm tr}({\bm {AB}}) = {\rm tr}({\bm {BA}})$ for matrices $\bm A$ and $\bm B$ of compatible dimensions, the signal power received by user $i \in {\mathcal U}$ can be expressed as ${{{{| \sum\nolimits_{j\in {\mathcal J}}{{f_{it}}{\bm h_{it}^{\rm T}}{\bm g_{ijt}}} |}^2}}} = {| {{\bm h}_{it}^{\rm T}{{\bm g}_{it}}} |^2} = {\left( {{\bm h}_{it}^{\rm T}{{\bm g}_{it}}} \right)^{\rm T}}{\bm h}_{it}^{\rm T}{{\bm g}_{it}} = {\rm tr}({\bm g}_{it}^{\rm T}{{\bm h}_{it}}{\bm h}_{it}^{\rm T}{{\bm g}_{it}}) = {\rm tr}({{\bm h}_{it}}{\bm h}_{it}^{\rm T}{{\bm g}_{it}}{\bm g}_{it}^{\rm T}) = {\rm tr}({{\bm H}_{it}}{{\bm G}_{it}})$.
Likewise, by introducing a square matrix $\bm Z_j \in {\mathbb R}^{JK \times JK}$ with $J \times J$ blocks, the transmit power for serving users can be written as $ {{\bm g}_{ijt}^{\rm T}{{\bm g}_{ijt}}}  = {\rm tr}({\bm Z}_j{{\bm G}_{it}})$. Besides, each block in $\bm Z_j$ is a $K \times K$ matrix. In $\bm Z_j$, the block in the $j$-th row and $j$-th column is a $K \times K$ identity matrix, and all other blocks are zero matrices. Then, by applying ${\bm G_{it}} = {\bm g_{it}}{\bm g}_{it}^{\rm T}$ $\Leftrightarrow {\bm G_{it}} \succeq 0$ and ${{\rm rank}({\bm G_{it}}) \le 1}$, we can convert (\ref{eq:original_problem_p0}) to the following problem
\begin{subequations}\label{eq:transformed_problem}
\begin{alignat}{2}
& \mathop {{\rm{maximize}}}\limits_{\{\bm a_{t}^{\rm ul}, \bm a_{t}^{\rm dl}, \bm p_t, \bm G_{it}\}} \text{ } \bar B(T) - \frac{1}{T}\sum\limits_{t = 1}^T {\sum\limits_{i \in {\mathcal U}} {\sum\limits_{j \in {\mathcal J}} {a_{ijt}^{{\rm{ul}}}p_{it}^{{\rm{tot}}}/{{\tilde p}_i}} } } \\
& {\rm s.t.} \quad {\log _2}\left( {1 + \frac{{a_{it}^{{\rm{dl}}}{\rm{tr}}({\bm H_{it}}{\bm G_{it}})}}{{{N_0}{W^{{\rm{dl}}}} + \sum\nolimits_{m \in {{\mathcal M}_{it}}} {a_{mt}^{{\rm{dl}}}{\rm{tr}}({\bm H_{mt}}{\bm G_{mt}})} }}} \right)  \ge {{\gamma ^{{\rm{th}}}}}/{{W^{{\rm{dl}}}}}, \forall i,t \\
& \qquad \sum\nolimits_{i \in {\mathcal U}} {a_{it}^{{\rm{dl}}}{\rm{tr}}({\bm Z_j}{\bm G_{it}})}  + {{\tilde E}_j} \le {E_j},\forall j,t \\
& \qquad {\bm G_{it}} \succeq 0,\forall i,t \\
& \qquad {\rm{rank}}({\bm G_{it}}) \le 1,\forall i,t \\
& \qquad \rm {(\ref{eq:SINR_condition}), (\ref{eq:original_problem_p0}b)-(\ref{eq:original_problem_p0}f).}
\end{alignat}
\end{subequations}

Like (\ref{eq:original_problem_p0}), (\ref{eq:transformed_problem}) is difficult to be directly solved; thus, we first vertically decompose it into the following two subproblems.
\begin{itemize}
  \item  Uplink optimization subproblem: The uplink optimization subproblem is formulated as
  \begin{subequations}\label{eq:uplink_subproblem}
    \begin{alignat}{2}
    & \mathop {{\rm{maximize}}}\limits_{\{\bm a_{t}^{\rm ul}, \bm p_t\}} \quad \frac{1}{T}\sum\limits_{t = 1}^T \left({B_t^{\rm ul}\left( {\bm a_t^{\rm ul}} \right) - \sum\limits_{i \in {\mathcal U}} {\sum\limits_{j \in {\mathcal J}} {a_{ijt}^{{\rm{ul}}}p_{it}^{{\rm{tot}}}/{{\tilde p}_i}} } }\right) \\
    & {\rm s.t.} \text{ } \quad
     \rm { (\ref{eq:SINR_condition}),(\ref{eq:original_problem_p0}b)-(\ref{eq:original_problem_p0}d),(\ref{eq:original_problem_p0}f).}
    \end{alignat}
    \end{subequations}
  \item  Downlink optimization subproblem: The downlink optimization subproblem can be formulated as follows
    \begin{subequations}\label{eq:downlink_subproblem}
    \begin{alignat}{2}
    & \mathop {{\rm{maximize}}}\limits_{\{\bm a_{t}^{\rm dl}, \bm G_{it}\}} \quad \frac{1}{T}\sum\limits_{t = 1}^T {B_t^{\rm dl}\left( {\bm a_t^{\rm dl}} \right)}  \\
    & {\rm s.t.} \text{ } \quad
    \rm { (\ref{eq:original_problem_p0}e),(\ref{eq:transformed_problem}b)-(\ref{eq:transformed_problem}e).}
    \end{alignat}
    \end{subequations}
\end{itemize}

Next, we propose to solve the two subproblems separately by exploring DRL approaches.

\subsection{Solution to the uplink optimization subproblem}
(\ref{eq:uplink_subproblem}) is confirmed to be a mixed-integer and sequence-dependent optimization subproblem.
Fig. \ref{fig:fig_DNN_framework} shows a DRL approach of solving (\ref{eq:uplink_subproblem}).
\begin{figure*}[!t]
\centering
\includegraphics[width=4.1 in]{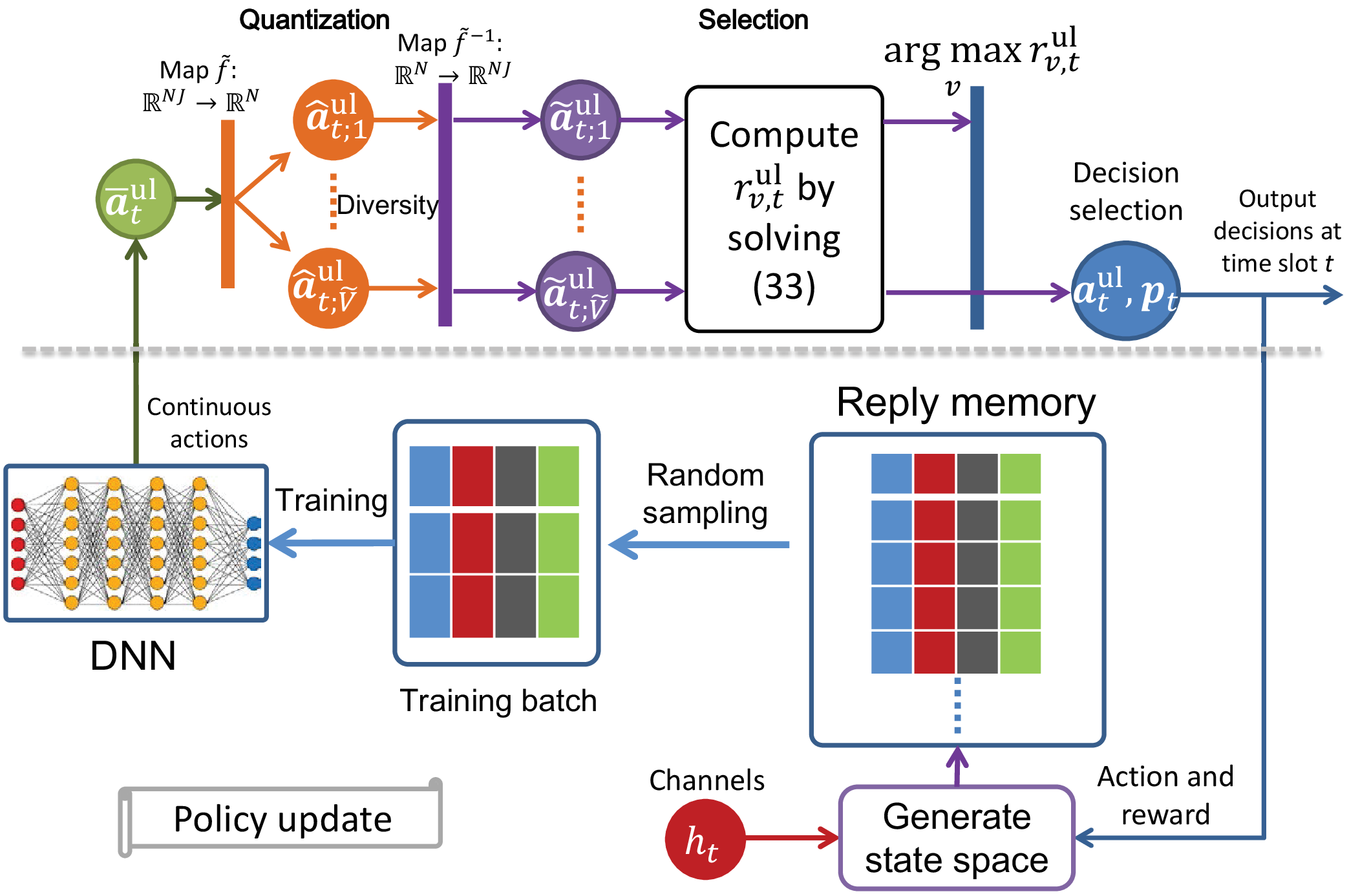}
\caption{A DRL approach of association and transmit power optimization.}
\label{fig:fig_DNN_framework}
\end{figure*}
In this figure, a DNN is trained to produce continuous actions.
The continuous actions are then quantized into a group of association (or discrete) actions.
Given the association actions, we solve an optimization problem to select an association action maximizing the reward.
Next, we describe the designing process of solving (\ref{eq:uplink_subproblem}) using a DRL-based uplink optimization method in detail.

\subsubsection{Action, state, and reward design}
First, we elaborate on the design of the state space, action space, and reward function of the DRL-based method. The HMDs' transmit power and the varying channel gains caused by users' movement and/or time-varying wireless channel environments have a significant impact on whether uplink transmission signals can be successfully decoded by APs. In addition, each AP has a limited ability to decode uplink transmission signals simultaneously.
Therefore, we design the state space, action space, and reward function of the DRL-based method as the following.
\begin{itemize}
\item \textbf{state space $\bm s_t^{\rm ul}$:} $\bm s_t^{\rm ul} = [{\bm m}_t; \hat {\bm h}_{t}^{\rm ul}; {\bm p}_t]$ is a column vector, where $m_{jt} \in {\bm m}_t \in {\mathbb R}^{J }$, $\forall j$, denotes the number of users successfully access to AP $j$ at time slot $t$. Besides, the state space involves the path-loss from user $i$ to AP $j$, $\hat h_{ijt} \in \hat {\bm h}_{t}^{\rm ul} \in {\mathbb R}^{NJ }$, $\forall i$, $j$, $t$, and the transmit power of user $i$'s HMD at time slot $t$, $p_{it} \in {\bm p}_t \in {\mathbb R}^{N }$, $\forall i$, $t$.
\item \textbf{action space $\bm a_{t}^{\rm ul}$:} $\bm a_t^{\rm ul} = [a_{11t}^{\rm ul}, \ldots, a_{1Jt}^{\rm ul}, \ldots, a_{NJt}^{\rm ul}]^{\rm T} \in {\mathbb R}^{NJ }$ with $a_{ijt}^{\rm ul} \in \{0, 1\}$. The action of the DRL-based method is to deliver users' data information to associated APs.
\item \textbf{reward $r_{t}^{\rm ul}$:} given $\bm a_{t}^{\rm ul}$, the reward $r_{t}^{\rm ul}$ is the objective function value of the following power control subproblem.
    \begin{subequations}\label{eq:power_allocation_uplink}
    \begin{alignat}{2}
    & r_{t}^{\rm ul} = \mathop {{\rm{maximize}}}\limits_{{ \bm p_t }} \text{ } B_{t}^{\rm ul}({\bm a}_t^{\rm ul})
            - \sum\nolimits_{i \in {\mathcal U}}{\sum\nolimits_{j\in{\mathcal J}}{a_{ijt}^{{\rm{ul}}}p_{it}^{{\rm{tot}}}/{{\tilde p}_i}}} \allowdisplaybreaks[4] \\
    & {\rm s.t.} \text{ } \quad \rm {(\ref{eq:SINR_condition}),(\ref{eq:original_problem_p0}f).}
    \end{alignat}
    \end{subequations}
\end{itemize}

\subsubsection{Training process of the DNN}
For the DNN module $\bar {\bm a}_{t}^{\rm ul} = \mu({\bm s}_t^{\rm ul}|\theta_t^{\mu})$ shown in Fig. \ref{fig:fig_DNN_framework}, where $\bar {\bm a}_{t}^{\rm ul} = [\bar {\bm a}_{1t}^{\rm ul};\ldots;\bar {\bm a}_{Nt}^{\rm ul}]$ and $\theta_t^{\mu}$ represents network parameters, we explore a two-layer fully-connected feedforward neural network with network parameters being initialized by a Xavier initialization scheme. There are $N_1^{\mu}$ and $N_2^{\mu}$ neurons in the $1^{\rm st}$ and $2^{\rm nd}$ hidden layers of the constructed DNN, respectively. Here, we adopt the ReLU function as the activation function in these hidden layers. For the output layer, a sigmoid activation function is leveraged such that relaxed association variables satisfy $\bar a_{ijt}^{\rm ul} \in (0, 1)$.
In the action-exploration phase, the exploration noise $\epsilon {N_f}$ is added to the output layer of the DNN, where $\epsilon \in (0, 1)$ decays over time and ${N_f} \sim {\mathcal N}(0, \sigma^2)$.

To train the DNN effectively, the experience replay technique is exploited. This is because there are two special characteristics in the process of enhancing users' fictitious experiences: 1) the collected input state values $\bm s_t^{\rm ul}$ incrementally arrive as users move to new positions, instead of all made available at the beginning of the training; 2) APs consecutively collect state values indicating that the collected state values may be closely correlated. The DNN may oscillate or diverge without breaking the correlation among the input state values.
Specifically, at each training epoch $t$, a new training sample $(\bm s_t^{\rm ul},  {\bm a}_t^{\rm ul}, \bm s_{t+1}^{\rm ul})$ is added to the replay memory. When the memory is filled, the newly generated sample replaces the oldest one. We randomly choose a minibatch of training samples $\{(\bm s_{\tau}^{\rm ul},  {\bm a}_{\tau}^{\rm ul}, \bm s_{\tau+1}^{\rm ul})| \tau \in {\mathcal T}_t\}$ from the replay memory, where ${\mathcal T}_t$ is a set of training epoch indices. The network parameters $\theta_t^{\mu}$ are trained using the ADAM method \cite{DBLP:journals/corr/KingmaB14} to reduce the averaged cross-entropy loss
\begin{equation}\label{eq:mse}
\begin{array}{l}
L(\theta_t^{\mu}) = - \frac{1}{{|{{\mathcal T}_t}|}}\sum\nolimits_{\tau  \in {{\mathcal T}_t}} ( {({\bm a}_{\tau}^{\rm ul})^{\rm T}} \log {\bar {\bm a}_{\tau}^{\rm ul}} + (1-{{\bm a}_{\tau}^{\rm ul}})^{\rm T} \log (1 - {\bar {\bm a}_{\tau}^{\rm ul}})  ).
\end{array}
\end{equation}

As evaluated in the simulation, we can train the DNN every $T_{ti}$ epochs after collecting a sufficient number of new data samples.

\subsubsection{Action quantization and selection method}
In the previous subsection, we design a continuous policy function and generate a continuous action space. However, a discrete action space is required in this paper. To this aim, the generated continuous action should be quantized, as shown in Fig. \ref{fig:fig_DNN_framework}.
A quantized action will directly determine the feasibility of the optimization subproblem and then the convergence performance of the DRL-based optimization method.
To improve the convergence performance, we should increase the diversity of the quantized action set, which including all quantized actions.
Specifically, we quantize the continuous action $\bar {\bm a}_t^{\rm ul}$ to obtain $\tilde V \in \{1, 2, \ldots, 2^N\}$ groups of association actions and denote by ${\bar {\bm a}_{t;v}^{\rm ul}}$ the $v$-th group of actions.
Given ${\bar {\bm a}_{it;v}^{\rm ul}}$, (\ref{eq:power_allocation_uplink}) is reduced to a linear programming problem, and we can derive its closed-form solution as below
\begin{equation}\label{eq:power_expression}
{p_{it}} = \left\{ {\begin{array}{*{20}{l}}
{\sum\limits_j {\frac{{a_{ijt}^{{\rm{ul}}}{\theta ^{{\rm{th}}}}{N_0}{W^{{\rm{ul}}}}}}{{N{f_i}{{\hat h}_{ijt}}}},} }&{\sum\limits_j {\frac{{a_{ijt}^{{\rm{ul}}}{\theta ^{{\rm{th}}}}{N_0}{W^{{\rm{ul}}}}}}{{N{f_i}{{\hat h}_{ijt}}}}}  \le {{\tilde p}_i} - p_i^c,}\\
{0,}&{{\rm otherwise.}}
\end{array}} \right.
\end{equation}

Besides, a great $\tilde V$ will result in higher diversity in the quantized action set but a higher computational complexity, and vice versa.
To balance the performance and complexity, we set $\tilde V=N$ and propose a lightweight action quantization and selection method. The detailed steps of quantizing and selecting association actions are given in Algorithm \ref{Alg:action_quantization}.
\begin{algorithm}
\caption{Action quantization and selection}
\label{Alg:action_quantization}
\begin{algorithmic}[1]
\STATE \textbf{Input:} The output action of the uplink DNN, ${\bar {\bm a}_{t}^{\rm ul}}$.
\STATE Arrange ${\bar {\bm a}_{t}^{\rm ul}}$ as a matrix of size $N \times J$ and generate a vector $\hat {\bm a}_t^{\rm ul} = \left\{ {\max [\bar a_{i1t}^{\rm ul}, \ldots ,\bar a_{iJt}^{\rm ul}],\forall i} \right\}$.
\STATE Generate the reference action vector $\bar {\bm b}_t = [\bar b_{1t}, \ldots, \bar b_{vt}, \ldots, \bar b_{\tilde Vt}]^{\rm T}$ by sorting the absolute value of all elements of $\hat {\bm a}_t^{\rm ul}$ in ascending order.
\STATE For any user $i$, generate the $1^{\rm st}$ group of association actions by
    \begin{equation}\label{eq:1st_binary_action}
        {\hat a_{it;1}^{\rm ul}} = \left\{ {\begin{array}{*{20}{c}}
            {1,}&{{\hat a_{it}^{\rm ul}} > 0.5,}\\
            {0,}&{{\hat a_{it}^{\rm ul}} \le 0.5.}
            \end{array}} \right.
    \end{equation}
\STATE For any user $i$, generate the remaining $\tilde V-1$ groups of association actions by
    \begin{equation}\label{eq:v_t_update}
        {\hat a_{it;v}^{\rm ul}} = \left\{ {\begin{array}{*{20}{c}}
            {1,}&{{\hat a_{it}^{\rm ul}} > {{\bar b}_{(v - 1)t}}, \ v=2,\ldots, \tilde V,}\\
            {0,}&{{\hat a_{it}^{\rm ul}} \le {{\bar b}_{(v - 1)t}}, \ v=2,\ldots, \tilde V. }
            \end{array}} \right.
    \end{equation}
\STATE {For each group of association actions $v \in \{1, 2, \ldots, \tilde V\}$, user $i$, and AP $j$, set}
    \begin{equation}\label{eq:v_t_update}
        \tilde a_{ijt;v}^{{\rm{ul}}} = \left\{ {\begin{array}{*{20}{c}}
        {1,}&{\hat a_{it;v}^{{\rm{ul}}} = 1,j = j^{\star},}\\
        {0,}&{\rm otherwise.}
        \end{array}} \right.
    \end{equation}
    where, $j^{\star} = \mathop {\arg \max }\limits_j [\bar a_{i1t}^{\rm ul}, \ldots ,\bar a_{iJt}^{\rm ul}]$.
\STATE {For each group of association actions $v \in \{1, 2, \ldots, \tilde V\}$, given the vector ${\tilde {\bm a}_{t;v}^{\rm ul}} = [\tilde a_{i1t;v}^{{\rm{ul}}},\ldots, \tilde a_{iJt;v}^{{\rm{ul}}}]_{i}^{\rm T}$, $\forall i$, solve (\ref{eq:power_allocation_uplink}) to obtain $r_{vt}^{\rm ul}$. }
\STATE Select the association action ${\bm a}_{t}^{\rm ul} = \arg \mathop {\max }\nolimits_{\{ \tilde a_{ijt;v}^{{\rm{ul}}}\} } r_{vt}^{{\rm{ul}}}$.
\STATE \textbf{Output:} The association action ${\bm a}_{t}^{\rm ul}$.
\end{algorithmic}
\end{algorithm}

Summarily, the proposed DRL-based uplink optimization method can be presented in Algorithm \ref{Alg:DRL_uplink_opt}.
\begin{algorithm}[!htp]
\caption{DRL-based uplink optimization}
\label{Alg:DRL_uplink_opt}
\begin{algorithmic}[1]
\STATE {\textbf {Initialize:}} The maximum number of episodes $N_{epi}$, the maximum number of epochs per episode $N_{epo}$, initial exploration decaying rate $\epsilon$, DNN $\mu(\bm s_t^{\rm ul}|\theta_t^{\mu})$ with network parameters $\theta_t^{\mu}$, initial reward $r_0^{\rm ul} = 1$, and users' randomly initialized transmit power.
\STATE {\textbf {Initialize:}} Replay memory with capacity $C$, minibatch size $|{{\mathcal T}_t}|$, and DNN training interval $T_{\rm ti}$.
\FOR{each episode in $\{1, \ldots, N_{epi}\}$}
\STATE Calculate the state space according to locations of APs and users and users' randomly initialized transmit power.
\FOR {each epoch $\bar t = 1, \ldots, N_{epo}$}
\STATE Select a relaxed action vector $\bar {\bm a}_{\bar t}^{\rm ul} = \mu({\bm s}_{\bar t}^{\rm ul}|\theta_{\bar t}^{\mu}) + \epsilon N_f$, where $\epsilon$ decays over time.
\STATE Call Algorithm \ref{Alg:action_quantization} to choose the association action $\bm a_{\bar t}^{\rm ul}$.
\IF {$\bm a_{\bar t}^{\rm ul}$ results in the violation of constraints in (\ref{eq:uplink_subproblem})}
\STATE {Cancel the action and update the reward by $r_{\bar t}^{\rm ul} = r_{\bar t}^{\rm ul} - \varpi  |r_{{\bar t}-1}^{\rm ul}|$.}
\ELSE
\STATE {Execute the action and observe the subsequent state ${\bm s}_{{\bar t}+1}^{\rm ul}$.}
\ENDIF
\STATE Store the transition $(\bm s_{\bar t}^{\rm ul}, \bm a_{\bar t}^{\rm ul}, \bm s_{{\bar t}+1}^{\rm ul})$ in the memory.
\STATE If ${\bar t} \ge |{{\mathcal T}_t}|$, sample a random minibatch of $|{{\mathcal T}_t}|$ transitions $(\bm s_m^{\rm ul}, \bm a_m^{\rm ul}, \bm s_{m+1}^{\rm ul})$ from the memory.
\STATE If $\bar t \mod T_{\rm ti} == 0$, update the network parameters ${\theta_{\bar t} ^{{\mu }}}$ by minimizing the loss function $L({\theta_{\bar t} ^{{\mu}}})$ using the ADAM method.
\ENDFOR
\ENDFOR
\end{algorithmic}
\end{algorithm}

\subsection{Solution to the downlink optimization subproblem}
Like (\ref{eq:uplink_subproblem}), (\ref{eq:downlink_subproblem}) is also a mixed-integer and sequence-dependent optimization problem. Therefore, the procedure of solving (\ref{eq:downlink_subproblem}) is similar to that of solving (\ref{eq:uplink_subproblem}), and we do not present the detailed steps of the DRL-based downlink optimization method in this subsection for brevity.
However, there are differences in some aspects, for example, the design of action and state space and the reward function. For the DRL-based downlink optimization method, we design its action space, state space, and the reward function as the following.
\begin{itemize}
\item \textbf{state space $\bm s_t^{\rm dl}$:} $\bm s_t^{\rm dl} = [\bm o_{t}; \bm h_t; {\bm I}_{t}^{\rm dl}; \bm g_t]$ is a column vector, where $o_{jt} \in {\bm o}_t \in {\mathbb R}^{J }$ indicates the number of users to which AP $j$ transmits VR video frames, $h_{ijkt} \in {\bm h}_t \in {\mathbb C}^{NJK }$, $I_{imt} \in {\mathbb R}^{N \times N} \in {\bm I}_{t}^{\rm dl}$ denotes whether user $m$ is the interfering user of user $i$, and $g_{ijkt} \in {\bm g}_t \in {\mathbb C}^{NJK }$.
\item \textbf{action space $\bm a_{t}^{\rm dl}$:} $\bm a_t^{\rm dl} = [a_{1t}^{\rm dl}, \ldots, a_{it}^{\rm dl}, \ldots, a_{Nt}^{\rm dl}]^{\rm T}$ with $a_{it}^{\rm dl} \in \{0, 1\}$. The action of the DRL-based method at time slot $t$ is to transmit VR video frames to corresponding users.
\item \textbf{reward $\bm r_{t}^{\rm dl}$:} given $\bm a_{t}^{\rm dl}$, the reward $r_{t}^{\rm dl}$ is the objective function value of the following power control subproblem.
    \begin{subequations}\label{eq:downlink_beamformer}
    \begin{alignat}{2}
    & r_{t}^{\rm dl} =  \mathop {{\rm{maximize}}}\limits_{{ {\bm G_{it}}} } {\text{  }} B_t^{{\rm{dl}}}\left( {\bm a_t^{{\rm{dl}}}} \right)\\
    & {\rm s.t.} \text{ } \quad {\rm (\ref{eq:transformed_problem}b)-(\ref{eq:transformed_problem}e)}.
    \end{alignat}
    \end{subequations}
\end{itemize}

To solve (\ref{eq:downlink_beamformer}), Algorithm \ref{Alg:action_quantization} can be adopted to obtain the downlink association action $\bm a_{t}^{\rm dl}$.
However, given $\bm a_{t}^{\rm dl}$, it is still hard to solve (\ref{eq:downlink_beamformer}) as (\ref{eq:downlink_beamformer}) is a non-convex programming problem with the existence of the non-convex low-rank constraint (\ref{eq:transformed_problem}e).
To handle the non-convexity, a semidefinite relaxation (SDR) scheme is exploited. The idea of the SDR scheme is to directly drop out the non-convex low-rank constraint. After dropping the constraint (\ref{eq:transformed_problem}e),
it can confirm that (\ref{eq:downlink_beamformer}) becomes a standard convex semidefinite programming (SDP) problem. This is because (\ref{eq:transformed_problem}b) are (\ref{eq:transformed_problem}c) are linear constraints w.r.t ${\bm G}_{it}$ and (\ref{eq:downlink_beamformer}a) is a constant objective function. We can then explore some optimization tools such as MOSEK to solve the standard convex SDP problem effectively.
However, owing to the relaxation, power matrices $\{\bm G_{it}\}$ obtained by mitigating (\ref{eq:downlink_beamformer}) without low-rank constraints will not satisfy the low-rank constraint in general. This is due to the fact that the (convex) feasible set of the relaxed (\ref{eq:downlink_beamformer}) is a superset of the (non-convex) feasible set of (\ref{eq:downlink_beamformer}).
The following lemma reveals the tightness of exploring the SDR scheme.
\begin{lemma}\label{lem:lemma_2}
\rm {
For any user $i$ at time slot $t$, denote by $\bm G_{it}^{\star}$ the solution to (\ref{eq:downlink_beamformer}). If ${\mathcal M}_{it} = \emptyset $, then the SDR for $\bm G_{it}$ in (\ref{eq:downlink_beamformer}) is tight, that is, ${\rm rank}(\bm G_{it}^{\star}) \le 1$; otherwise, we can not claim ${\rm rank}(\bm G_{it}^{\star}) \le 1$.
}
\end{lemma}
\begin{proof}
The Karush-Kuhn-Tucker (KKT) conditions can be explored to prove the tightness of resorting to the SDR scheme. Nevertheless, we omit the detailed proof for brevity as a similar proof can be found in Appendix of the work \cite{yang2020should}.
\end{proof}

With the conclusion in Lemma \ref{lem:lemma_2}, we can recover beamformers from the obtained power matrices. If ${\rm rank}(\bm G_{it}^{\star}) \le 1$, $\forall i$, then execute eigenvalue decomposition on $\bm G_{it}^{\star}$, and the principal component is the optimal beamformer $\bm g_{it}^{\star}$; otherwise, some manipulations such as a randomization/scale scheme \cite{luo2010semidefinite} should be performed on $\bm G_{it}^{\star}$ to impose the low-rank constraint.

Note that (\ref{eq:downlink_beamformer}) should be solved for $\tilde V$ times at each time slot. To speed up the computation, they can be optimized in parallel. Moreover, it is tolerable to complete the computation within the interval $[t, t+M]$ as users' locations in $M$ time slots are obtained.

Finally, we can summarize the DRL-based optimization algorithm of mitigating the problem of enhancing users' VR experiences in Algorithm \ref{alg:final_algorithm}.

\begin{algorithm}
\caption{DRL-based optimization algorithm}
\label{alg:final_algorithm}
\begin{algorithmic}[1]
\STATE \textbf{Initialization:} Run initialization steps of Algorithms \ref{Alg:Alg1}, \ref{Alg:action_quantization}, and \ref{Alg:DRL_uplink_opt}, and initialize the ESN training interval $T_{\rm pr}$.
\STATE Call Algorithm \ref{Alg:DRL_uplink_opt} to pre-train the uplink DNN $\mu({\bm s}_t^{\rm ul}|\theta_t^{\mu})$. Likewise, pre-train the downlink DNN $\mu({\bm s}_t^{\rm dl}|\theta_t^{Q})$.
\STATE Run steps 2-8 of Algorithm \ref{Alg:Alg1} to pre-train ESN models.
\FOR {each time slot $t = 1, 2, \ldots, T$}
\STATE Run step 9 of Algorithm \ref{Alg:Alg1} to obtain predicted location $\hat {\bm y}_{i(t+M)}$ of each user $i$.
\STATE Run steps 6-12 of Algorithm \ref{Alg:DRL_uplink_opt} to obtain uplink association action ${\bm a}_{t+M}^{\rm ul}$ and transmit power ${\bm p}_{t+M}$. Likewise, optimize the downlink association action ${\bm a}_{t+M}^{\rm dl}$ and transmit beamformer ${\bm g}_{i(t+M)}$ for each user $i$.
\IF {$ t \mod T_{\rm pr} == 0$}
\STATE Steps 2-8 of Algorithm \ref{Alg:Alg1}.
\ENDIF
\ENDFOR
\end{algorithmic}
\end{algorithm}

\section{Simulation and Performance Evaluation}
\subsection{Comparison algorithms and parameter setting}
To verify the effectiveness of the proposed algorithm, we compare it with three benchmark algorithms: 1) $k$-nearest neighbors (KNN) based action quantization algorithm: The unique difference between the KNN-based algorithm and the proposed algorithm lies in the scheme of quantizing uplink and downlink action spaces. For the KNN-based algorithm, it adopts the KNN method \cite{DBLP:journals/tmc/HuangBZ20} to quantize both uplink and downlink action spaces;
2) DROO algorithm: Different from the proposed algorithm, DROO leverages the order-preserving quantization method \cite{DBLP:journals/tmc/HuangBZ20} to quantize both uplink and downlink action spaces;
3) Heuristic algorithm: The heuristic algorithm leverages the greedy admission algorithm in \cite{tang2019service} to determine $\bm a_t^{\rm ul}$ and $\bm a_{t}^{\rm dl}$ at each time slot $t$. Besides, the user consuming less power in this algorithm will establish the connection with an AP(s) on priority.

To test the practicality of the developed parallel ESN learning method, realistic user movement datasets are generated via Google Map. Particularly, for a user, we randomly select its starting position and ending position on the campus of Singapore University of Technology and Design (SUTD). Given two endpoints, we use Google Map to generate the user's 2D trajectory. Next, we linearly zoom all $N$ users' trajectories into the communication area of size $0.5 \times 0.5$ km$^2$.

Additionally, the parameters related to APs and downlink transmission channels are listed as follows: the number of APs $J = 3$, the number of antenna elements $K = 2$, the antenna gain $G = 5$ dB, $g = 1$ dB, $\phi = \pi/3$, $\vartheta = \pi/2$, $W^{\rm dl} = 800$ MHz, $\gamma^{\rm th} = 1$ Gb/s, $\eta_{\rm LoS} = 2.0$, $\eta_{\rm NLoS} = 2.4$, $\sigma_{\rm LoS}^2 = 5.3$, $\sigma_{\rm NLoS}^2 = 5.27$, $D^{\rm th} = 50$ m, $x_o = y_o = 250$ m, $\theta_j = \pi/3$, $\tilde E_j = 40$ dBm, $E_j^c = 30$ dBm, $H_j = 5.5$ m, $\forall j$ \cite{DBLP:journals/twc/ChenSSLY20}.
User and uplink transmission channel-related parameters are shown as below: uplink system bandwidth $W^{\rm ul} = 200$ MHz, $\theta^{\rm th } = 200$, $\bar h = 1.8$ m, $\sigma_h^2 = 0.05$ m, $\alpha = 5$, $c_{ij} = 0.3$, $p_{i}^c = 23$ dBm, $\tilde p_i = 27$ dBm, $\forall i$, $j$.

Set other learning-correlated parameters as below: ${\zeta}  = 1$, $\xi = 0.25$, $\bar r_{\rm max} = 1000$, the sample number $Q = 6$, the number of future time slots $M = 8$, $N_i = 2$, $\forall i$, $N_o = 2$, $N_r = 300$, and $T_{\rm pr} = 5$.
For both uplink DNN and downlink DNN, the first hidden layer has $120$ neurons, and the second hidden layer has $80$ neurons.
The replay memory capacity $C = 1e$+6, $N_{epi} = 10$, $N_{epo} = 1000$, $\varpi = 10$, $\sigma^2 = 0.36$, $\epsilon = 0.99$.
More system parameters are listed as follows: carrier frequency $f_c = 28$ GHz, light of speed {$c = 3.0e$+8 m/s}, noise power spectral density $N_0 = -167$ dBm/Hz, and $T = 5000$ time slots.

\subsection{Performance evaluation}
To comprehensively understand the accuracy and the availability of the developed learning and optimization methods, we illustrate their performance results. In this simulation, we first let the AP number $J = 3$ and the mobile user number $N = 16$.

To validate the accuracy of the parallel ESN learning method on predicting mobile users' locations, we plot the actual trajectory of a randomly selected mobile user and its correspondingly predicted trajectory in Fig. \ref{fig_predicted_track}. In Fig. \ref{fig_nrmse}, the accuracy, which is measured by the normalized root mean-squared error (NRMSE) \cite{DBLP:journals/nn/ScardapaneWP16}, of predicted trajectories of $16$ mobile users is plotted.
From Fig. \ref{fig_x_y_coordinate}, we can observe that:
i) when the orientation angles of users will not change fast, the learning method can exactly predict users' locations. When users change their moving directions quickly, the method loses their true trajectories. However, the method will re-capture users' tracks after training ESN models based on newly collected users' location samples;
ii) the obtained NRMSE of the predicted trajectories of all mobile users will not be greater than $0.03$. Therefore, we may conclude that the developed parallel ESN learning method can be utilized to predict mobile users' locations.
\begin{figure}[!t]
\centering
\subfigure[A user's true and predicted trajectories]{\includegraphics[width=2.4in, height=1.3in]{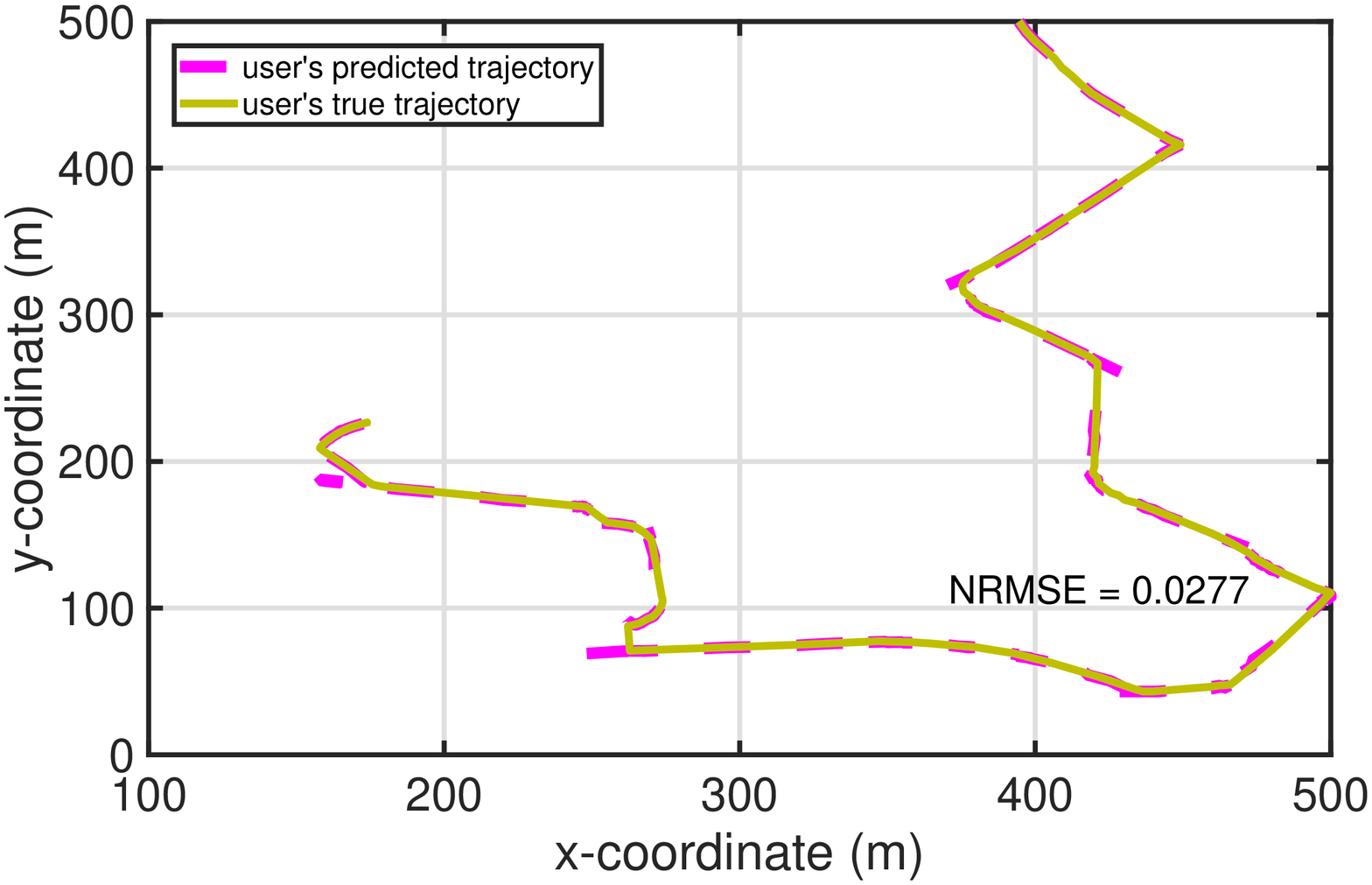}%
\label{fig_predicted_track}}
\hfil
\subfigure[NRMSE of predicted trajectories of $N$ users]{\includegraphics[width=2.4in, height=1.3in]{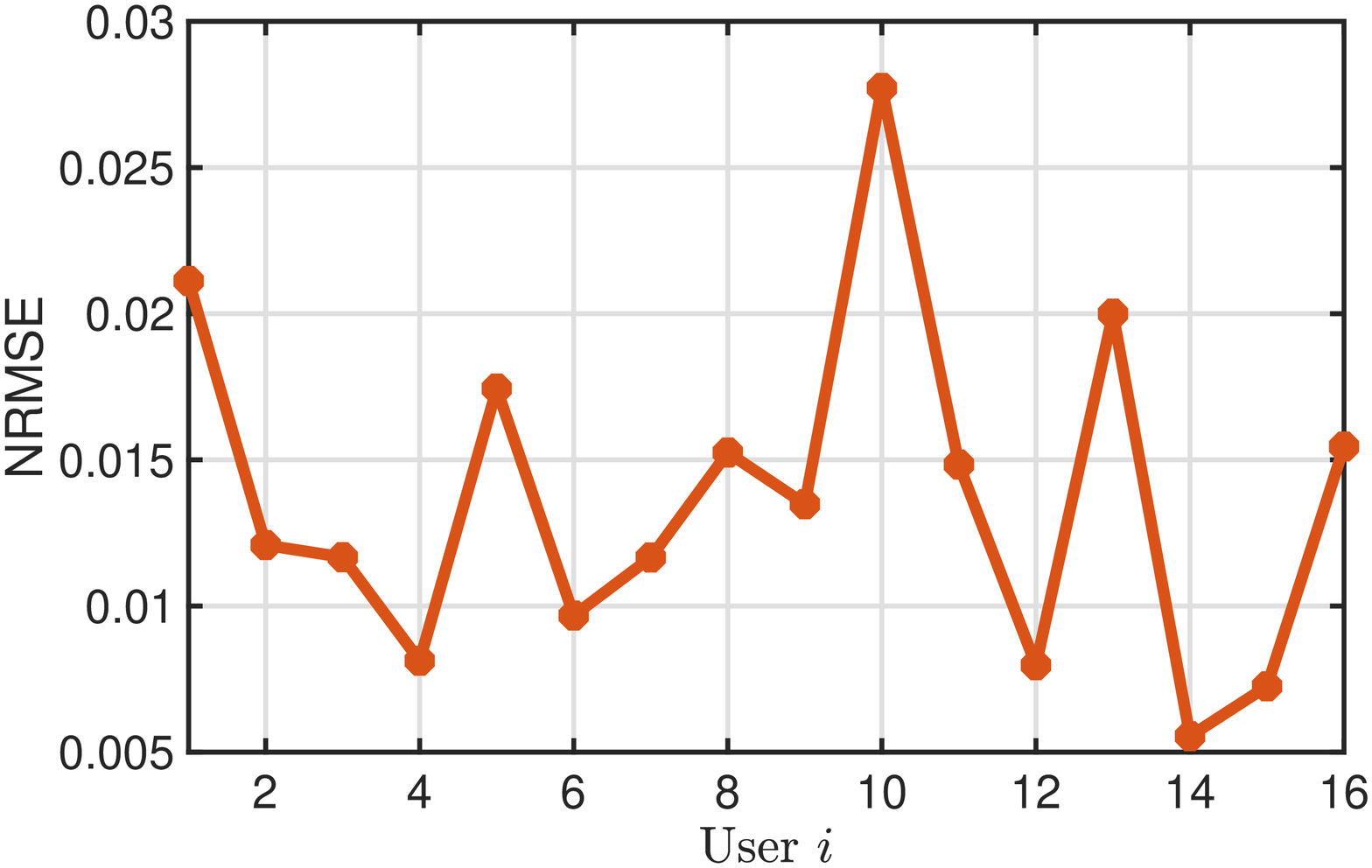}%
\label{fig_nrmse}}
\caption{Prediction accuracy of the parallel ESN learning method.}
\label{fig_x_y_coordinate}
\end{figure}

Next, to evaluate the performance of the proposed DRL-based optimization algorithm comprehensively, we illustrate the impact of some DRL-related crucial parameters such as minibatch size, training interval, and learning rate on the convergence performance of the proposed algorithm. DNN training loss and moving average reward, which is the average of the achieved rewards over the last $50$ epochs, are leveraged as the evaluation indicators.

Fig. \ref{fig_converge_batchsize} plots the tendency of the DNN training loss and the achieved moving average reward of the proposed algorithm under diverse minibatch sizes.
This figure illustrates that: i) a great minibatch size value will cause the DNN to converge slowly or even not.
As shown in Fig. \ref{fig_first_case_bats}, $L(\theta_{465}^{\mu}) = 0.1885$ when we set $|{\mathcal T}_{t}| = 512$. Yet, $L(\theta_{465}^{\mu}) = 0.1023$ when we let $|{\mathcal T}_{t}| = 64$. The result in Fig. \ref{fig_second_case_bats} shows that DNN does not converge after 10000 epochs when $|{\mathcal T}_{t}| = 2048$.
This is because a great $|{\mathcal T}_{t}|$ indicates overtraining, resulting in the local minima and degraded convergence performance.
Further, a large minibatch size value consumes more training time at each training epoch. Therefore, we set the training minibatch size $|{\mathcal T}_t| = 64$ in the simulation;
ii) when $|{\mathcal T}_t| = 64$, $r_{\bar t}^{\rm ul}$ and $r_{\bar t}^{\rm dl}$ gradually increase and stabilize at around $0.7141$ and $0.9375$, respectively. The fluctuation is mainly caused by the random sampling of training data and user movement.
\begin{figure*}[!t]
\centering
\subfigure[Uplink training loss vs. minibatch size]{\includegraphics[width=2.3in, height=1.3in]{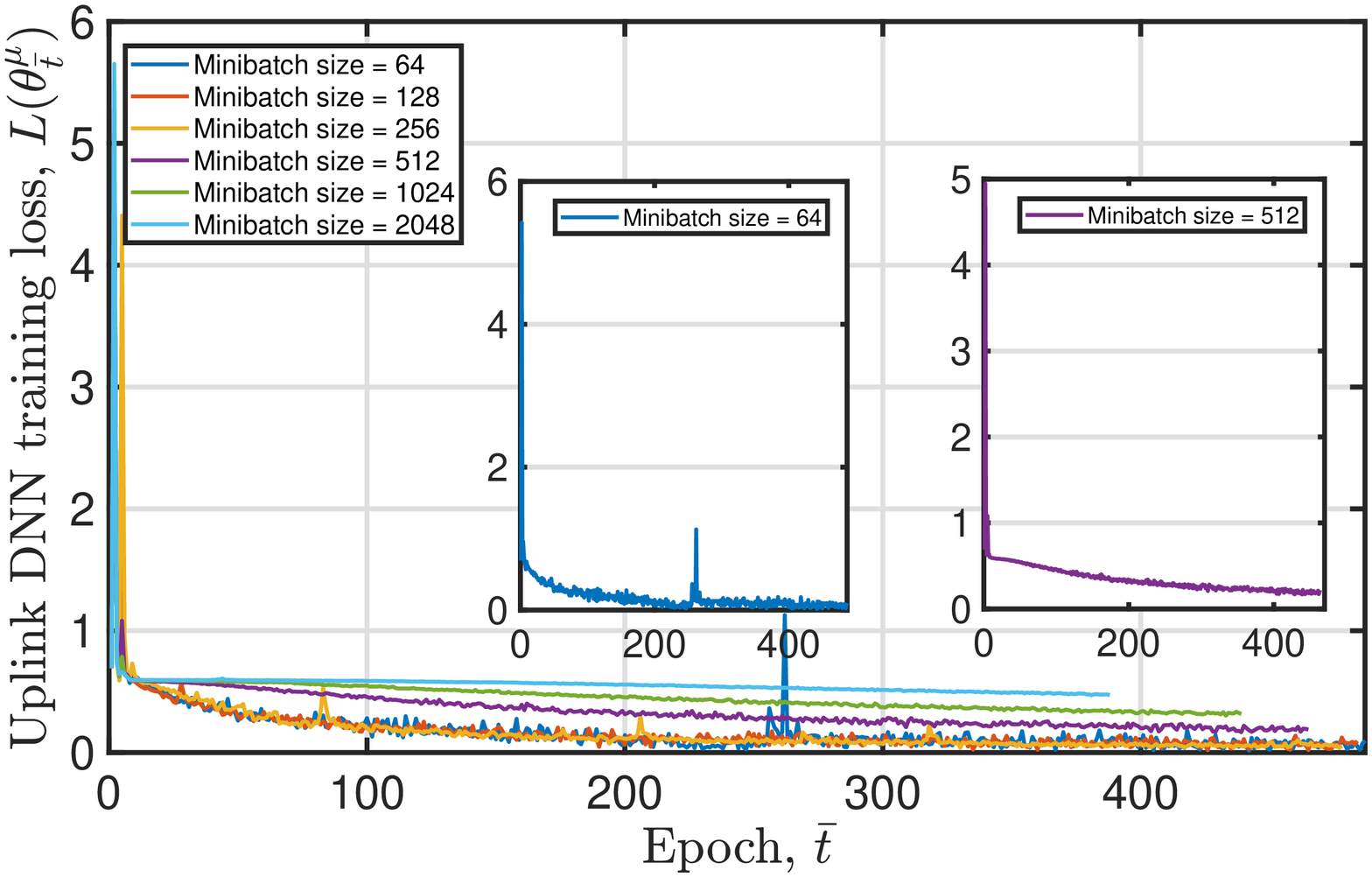}%
\label{fig_first_case_bats}}
\hspace{0.1\linewidth}
\subfigure[Uplink reward vs. minibatch size]{\includegraphics[width=2.3in, height=1.3in]{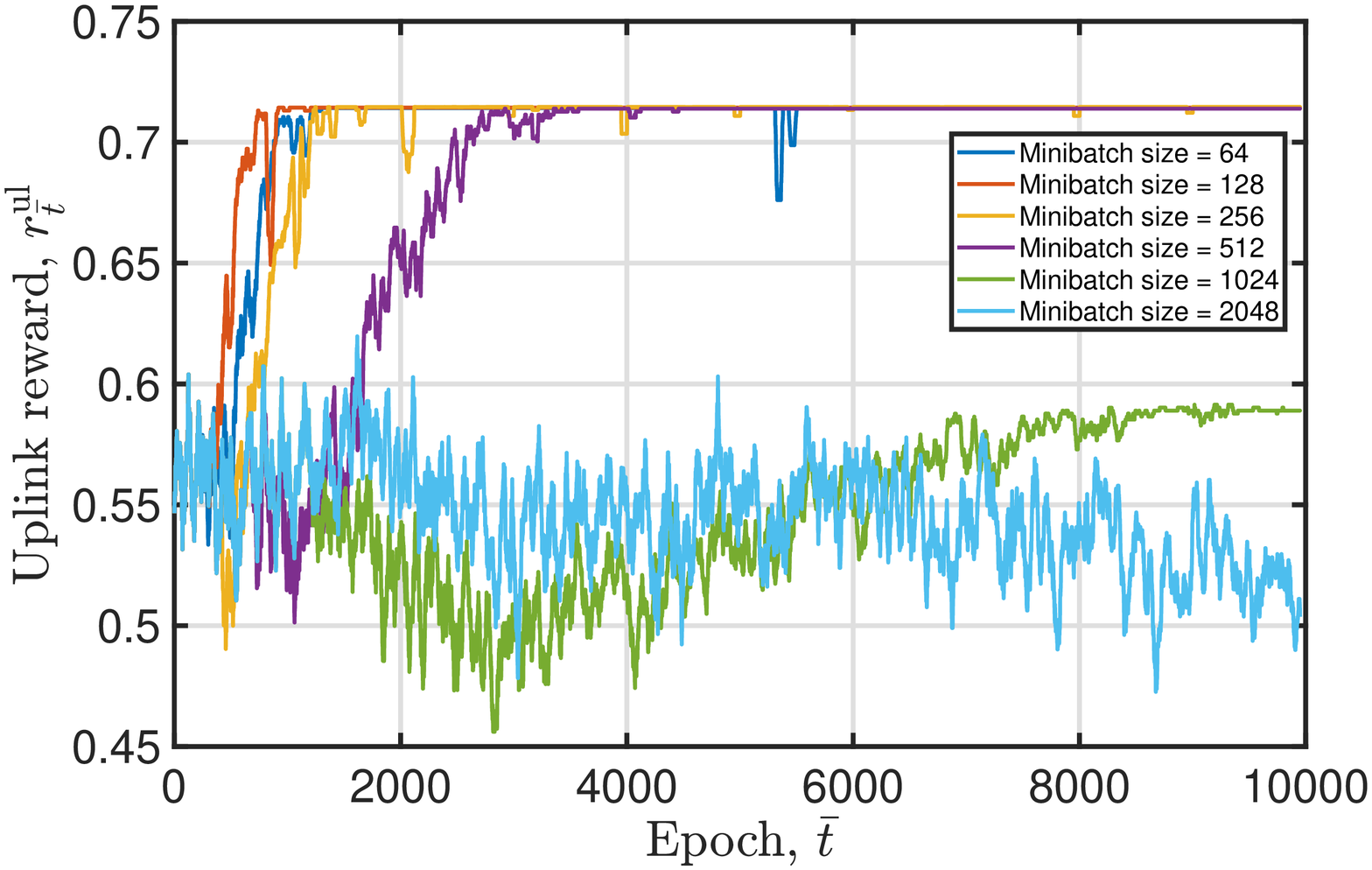}%
\label{fig_second_case_bats}}
\hfil
\subfigure[Downlink training loss vs. minibatch size]{\includegraphics[width=2.3in, height=1.3in]{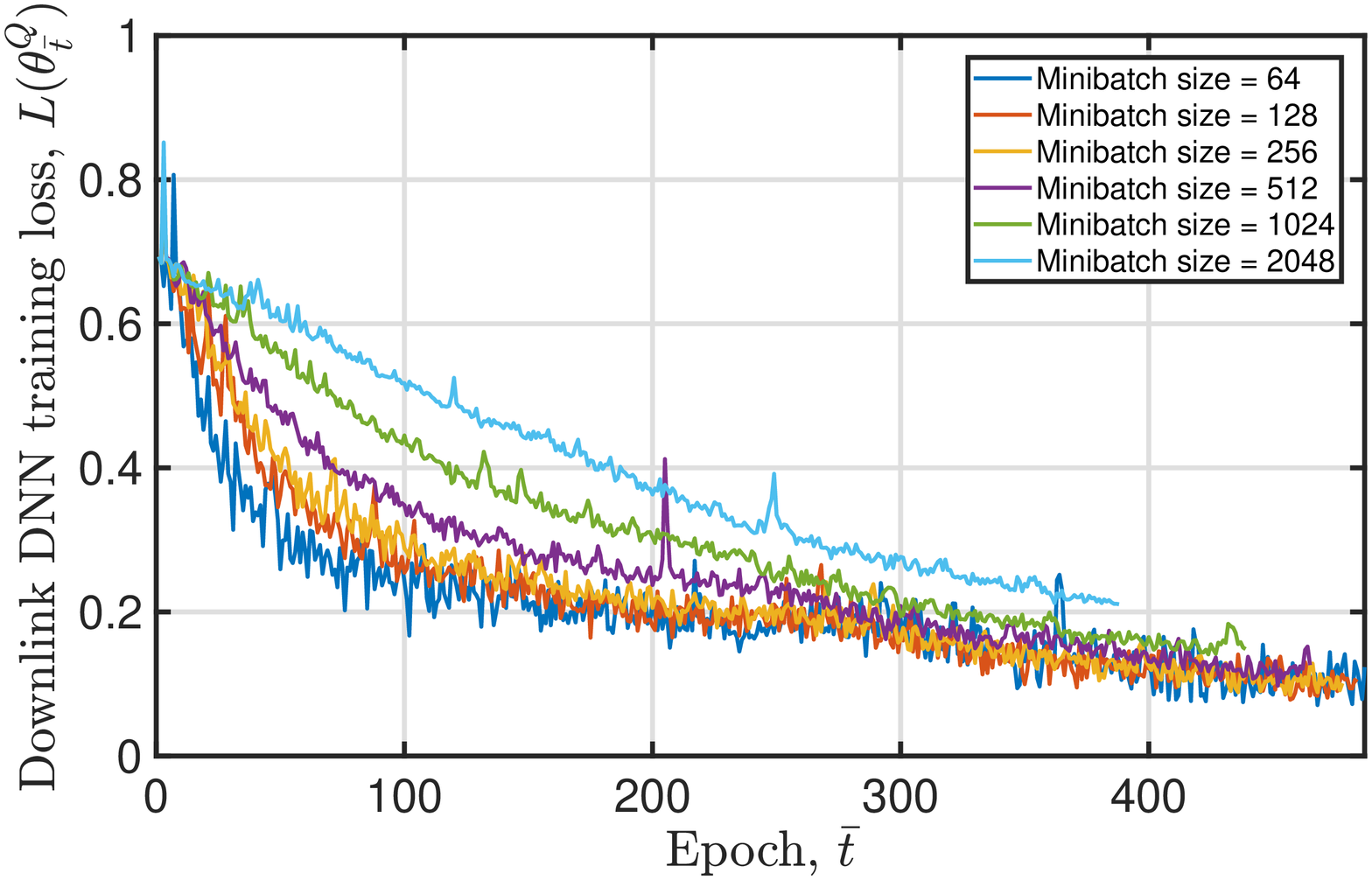}%
\label{fig_third_case_bats}}
\hspace{0.1\linewidth}
\subfigure[Downlink reward vs. minibatch size]{\includegraphics[width=2.3in, height=1.3in]{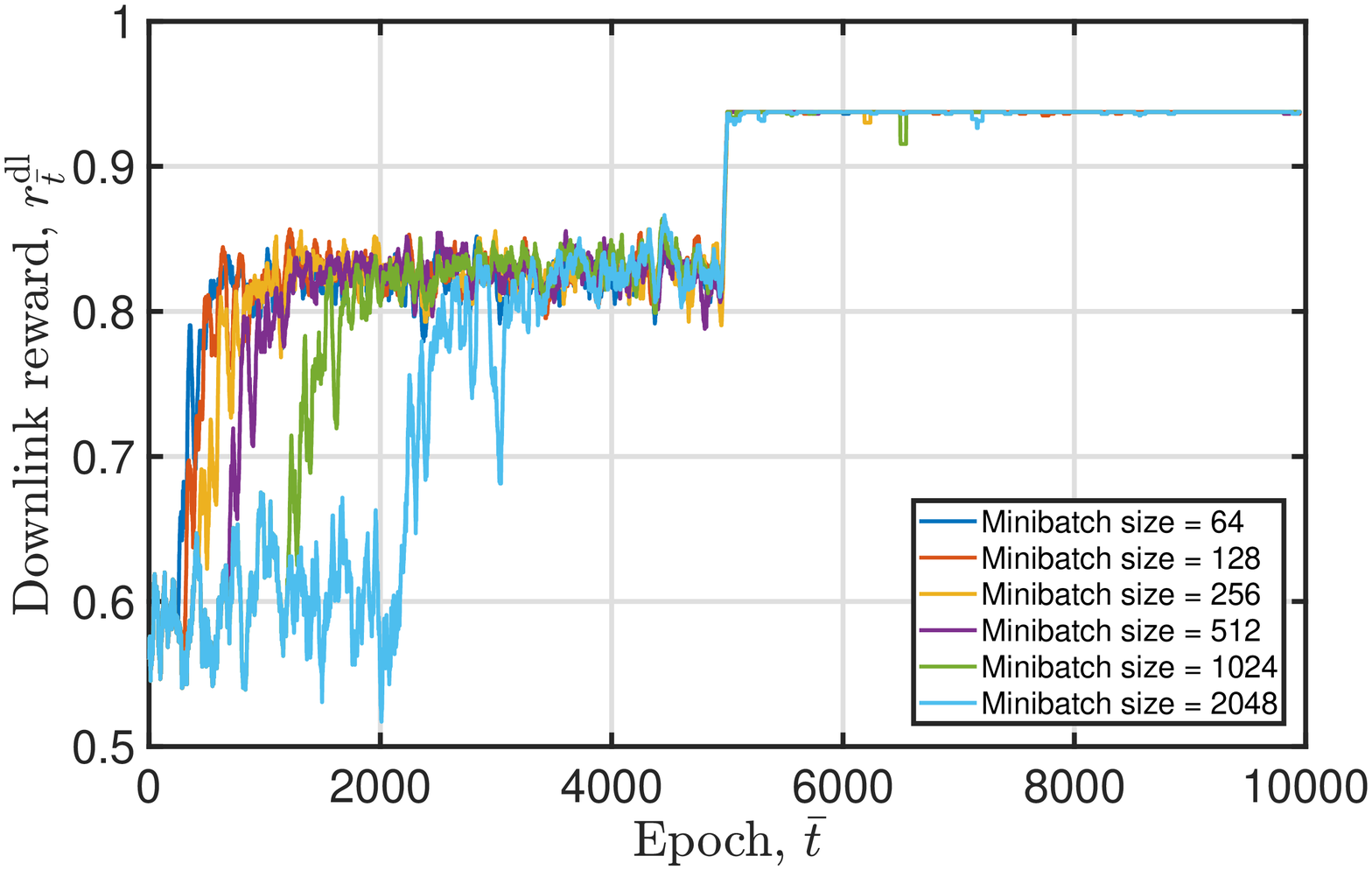}%
\label{fig_forth_case_bats}}
\caption{The impact of minibatch size $|{\mathcal T}_{t}|$ on the convergence performance of the proposed algorithm.}
\label{fig_converge_batchsize}
\end{figure*}

Fig. \ref{fig_converge_training_interval} illustrates the tendency of obtained uplink and downlink DNN training losses and moving average rewards under diverse training interval values.
From this figure, we can observe that a small training interval value indicates faster convergence speed. For example, if we set the training interval ${T_{\rm ti}} = 5$, the obtained $r_{\bar t}^{\rm ul}$ converges to $0.7156$ when epoch $\bar t > 439$. If we let the training interval ${T_{\rm ti}} = 100$,  $r_{\bar t}^{\rm ul}$ converges to $0.7149$ when epoch $\bar t > 4975$, as shown in Fig. \ref{fig_second_case_tri}.
However, it is unnecessary to train and update the DNN frequently, which will bring more frequent policy updates, if the DNN can converge. Thus, to achieve the trade-off between the convergence speed and the policy update speed, we set ${T_{\rm ti}} = 20$ in the simulation.
\begin{figure*}[!t]
\centering
\subfigure[Uplink training loss vs. training interval]{\includegraphics[width=2.4in, height=1.3in]{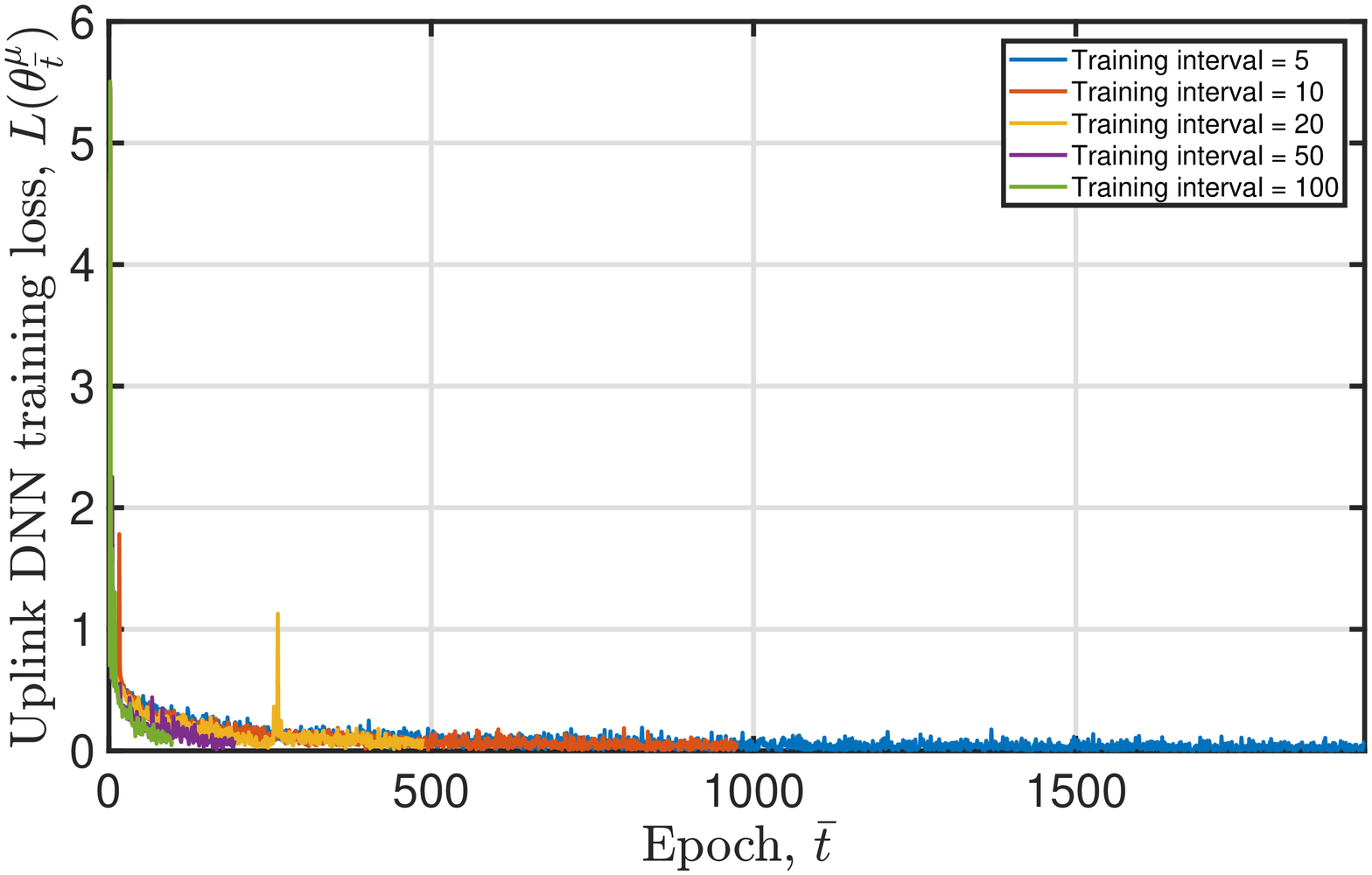}%
\label{fig_first_case_tri}}
\hspace{0.1\linewidth}
\subfigure[Uplink reward vs. training interval]{\includegraphics[width=2.4in, height=1.3in]{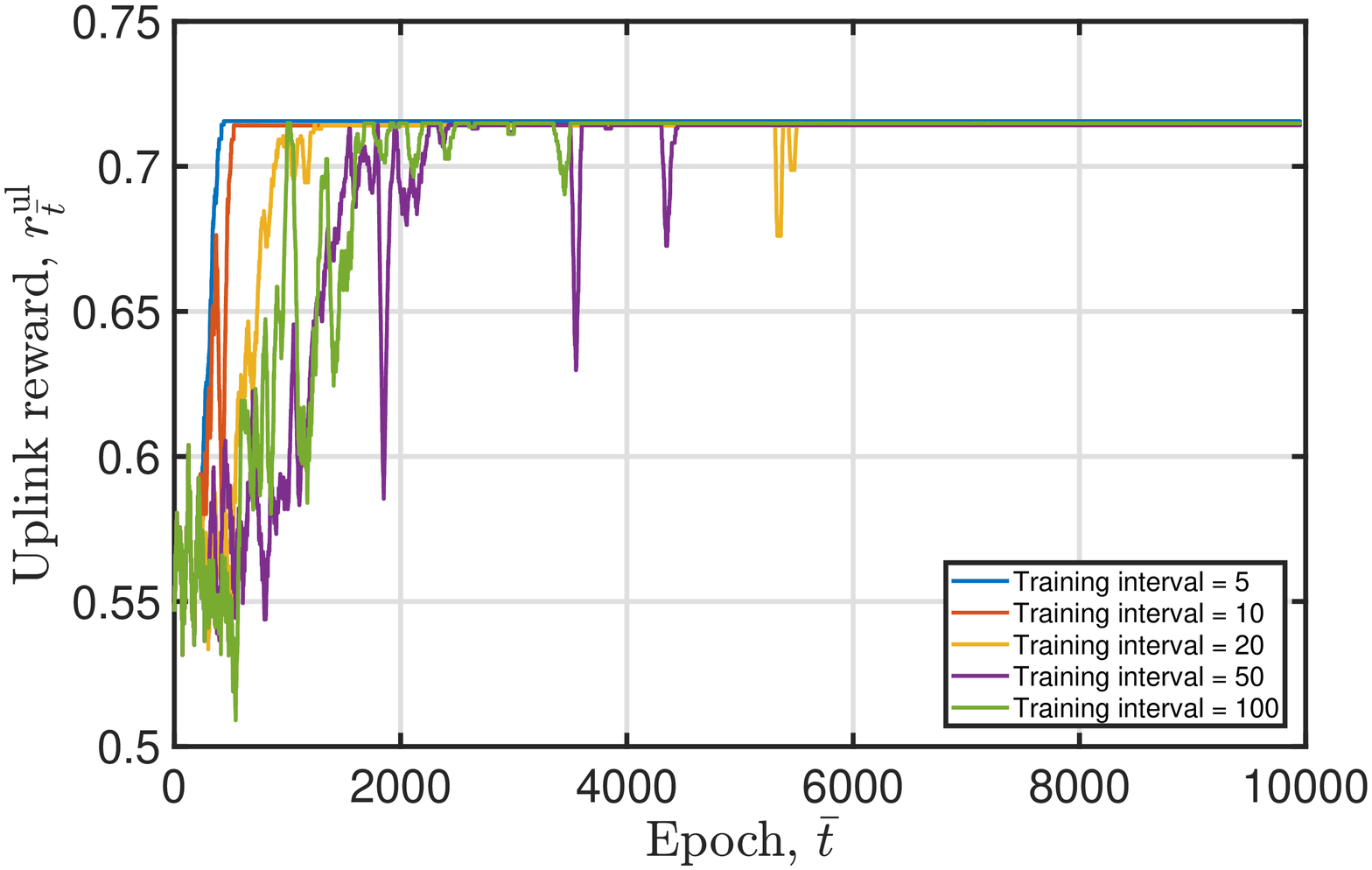}%
\label{fig_second_case_tri}}
\hfil
\subfigure[Downlink training loss vs. training interval]{\includegraphics[width=2.4in, height=1.3in]{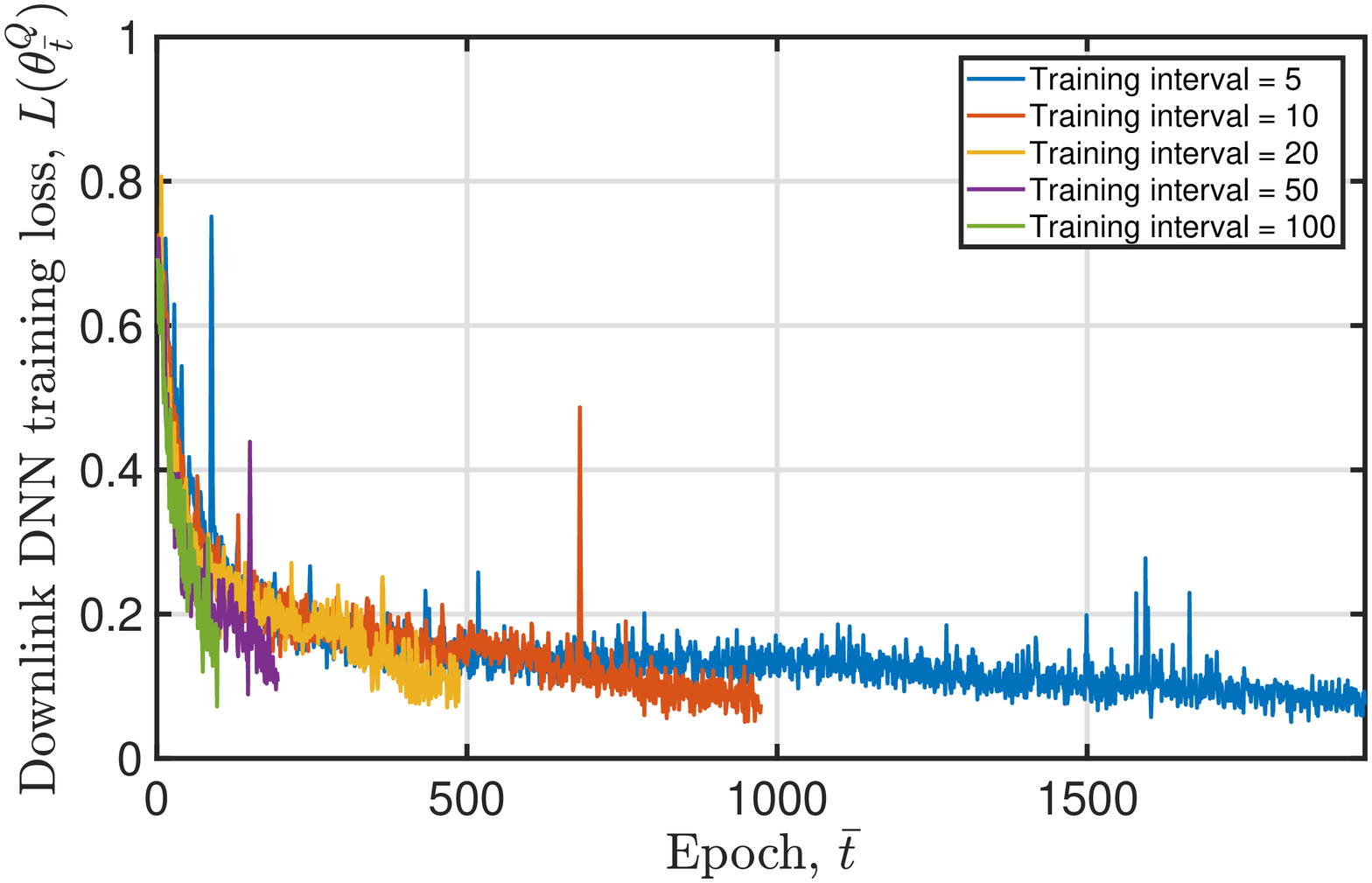}%
\label{fig_third_case_tri}}
\hspace{0.1\linewidth}
\subfigure[Downlink reward vs. training interval]{\includegraphics[width=2.4in, height=1.3in]{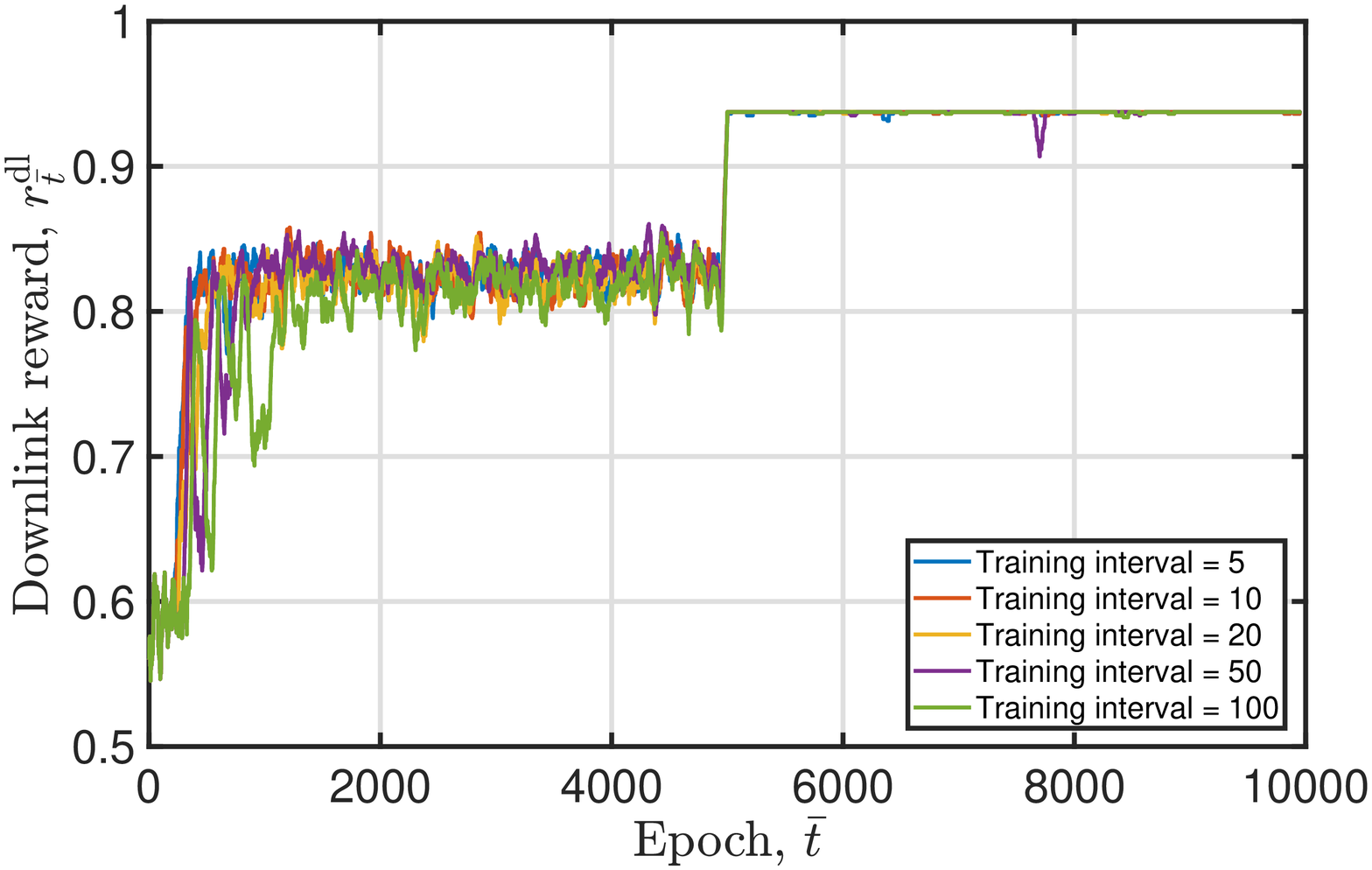}%
\label{fig_forth_case_tri}}
\caption{The impact of DNN training interval $T_{\rm ti}$ on the convergence performance of the proposed algorithm.}
\label{fig_converge_training_interval}
\end{figure*}

Fig. \ref{fig_converge_learning_rate} depicts the tendency of achieved DNN training loss and moving average reward of the proposed algorithm under different learning rate configurations.
From this figure, we have the following observations:
i) for the uplink DNN, when given a small learning rate value, it may converge to the local optimum or even not;
ii) for the downlink DNN, both a small and a great learning rate value will degrade convergence performance.
Therefore, when training the uplink DNN, we set the learning rate $l_r^{\rm ul} = 0.1$, which can lead to good convergence performance. For instance, $r_{\bar t}^{\rm ul}$ converges to $0.7141$ when epoch $\bar t \ge 1300$ and the variance of $r_{\bar t}^{\rm ul}$ gradually decreases to zero with an increasing epoch $\bar t$. We set the learning rate $l_r^{\rm dl} = 0.01$ when training the downlink DNN. Given this parameter setting, the obtained $L(\theta_{\bar t}^{Q})$ is smaller than $0.2$ after training for $200$ epochs.
\begin{figure*}[!t]
\centering
\subfigure[Uplink training loss vs. learning rate]{\includegraphics[width=2.3in, height=1.3in]{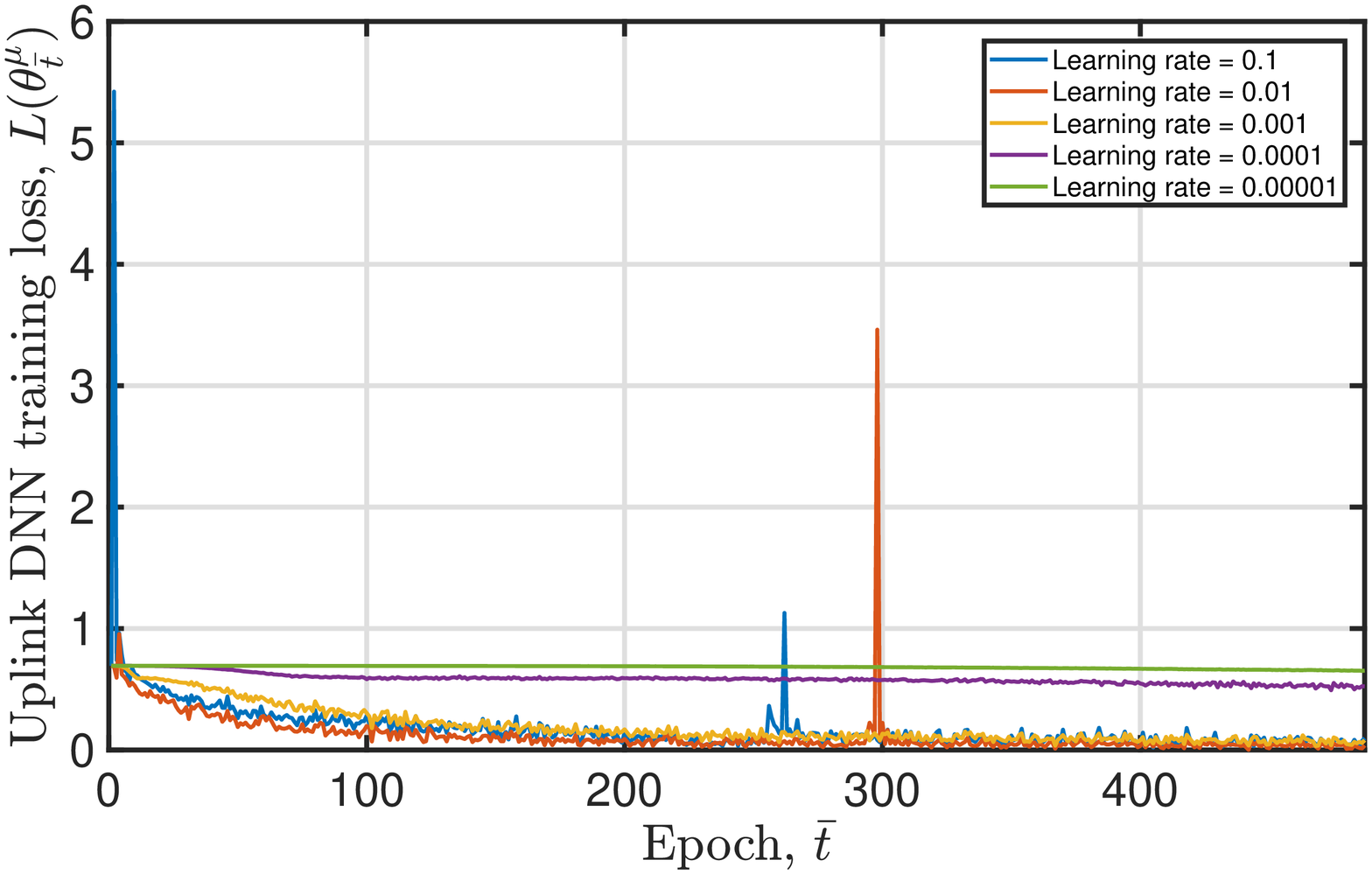}%
\label{fig_first_case_lr}}
\hspace{0.1\linewidth}
\subfigure[Uplink reward vs. learning rate]{\includegraphics[width=2.3in, height=1.3in]{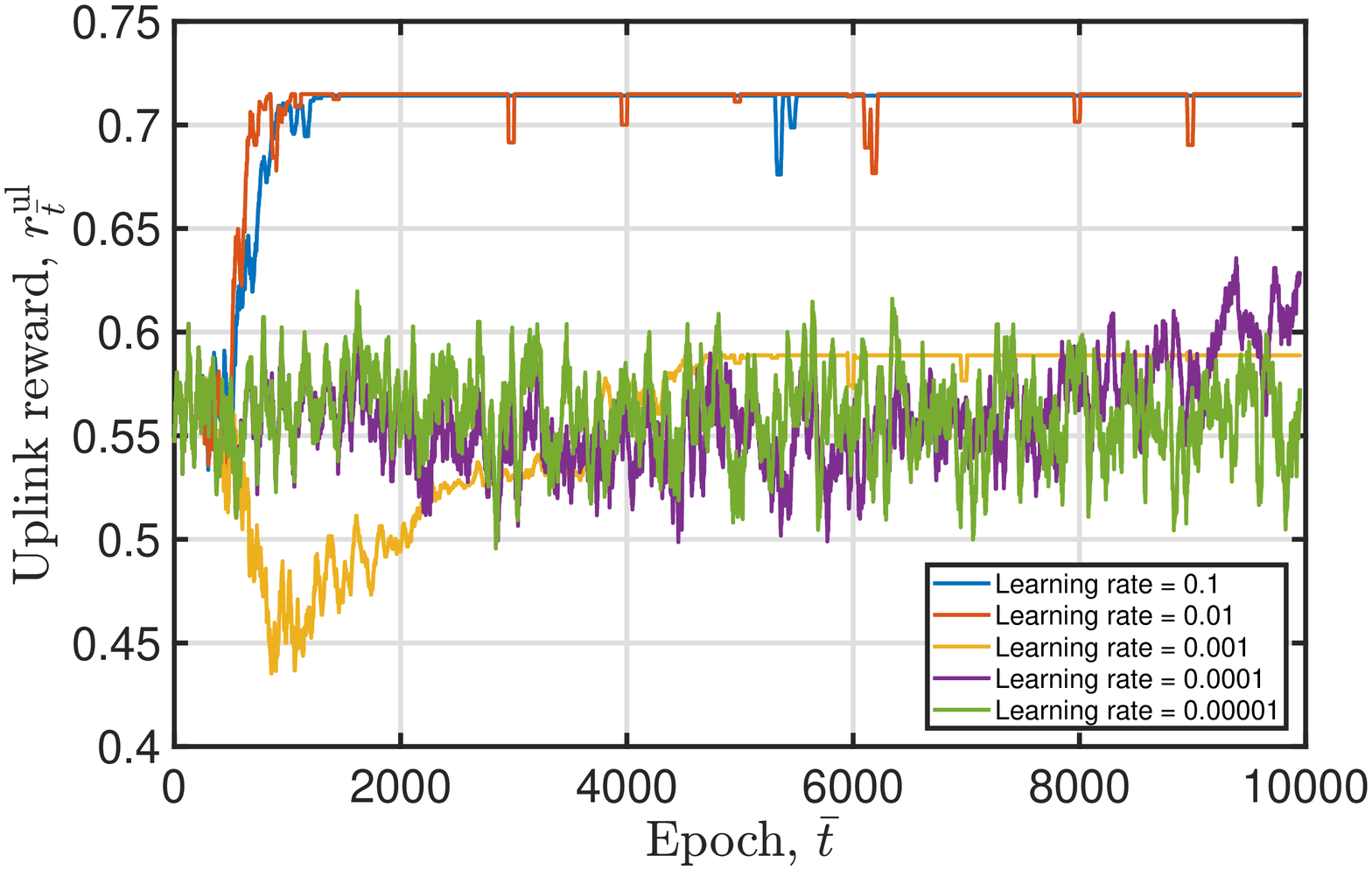}%
\label{fig_second_case_lr}}
\hfil
\subfigure[Downlink training loss vs. learning rate]{\includegraphics[width=2.3in, height=1.3in]{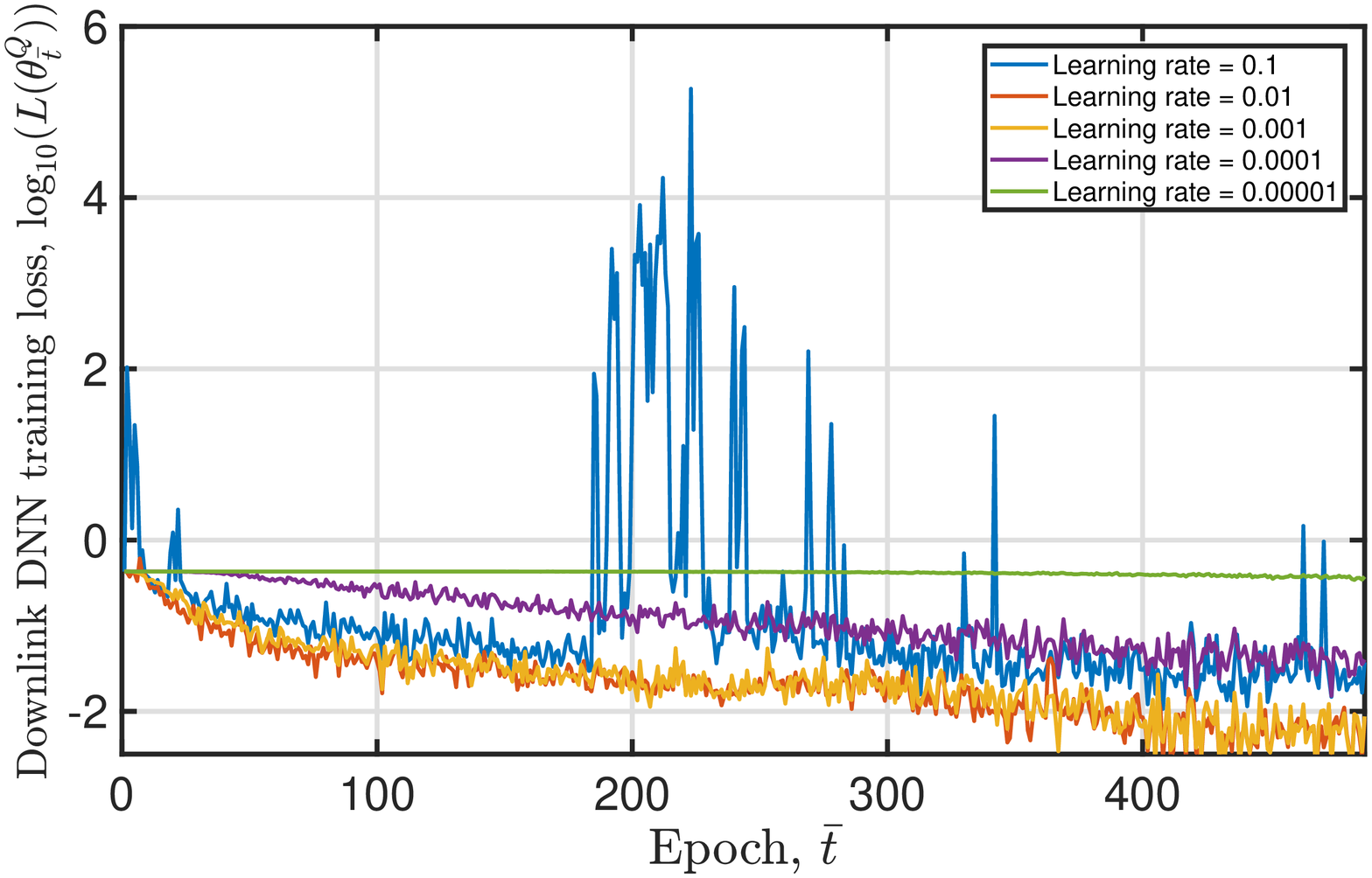}%
\label{fig_third_case_lr}}
\hspace{0.1\linewidth}
\subfigure[Downlink reward vs. learning rate]{\includegraphics[width=2.3in, height=1.3in]{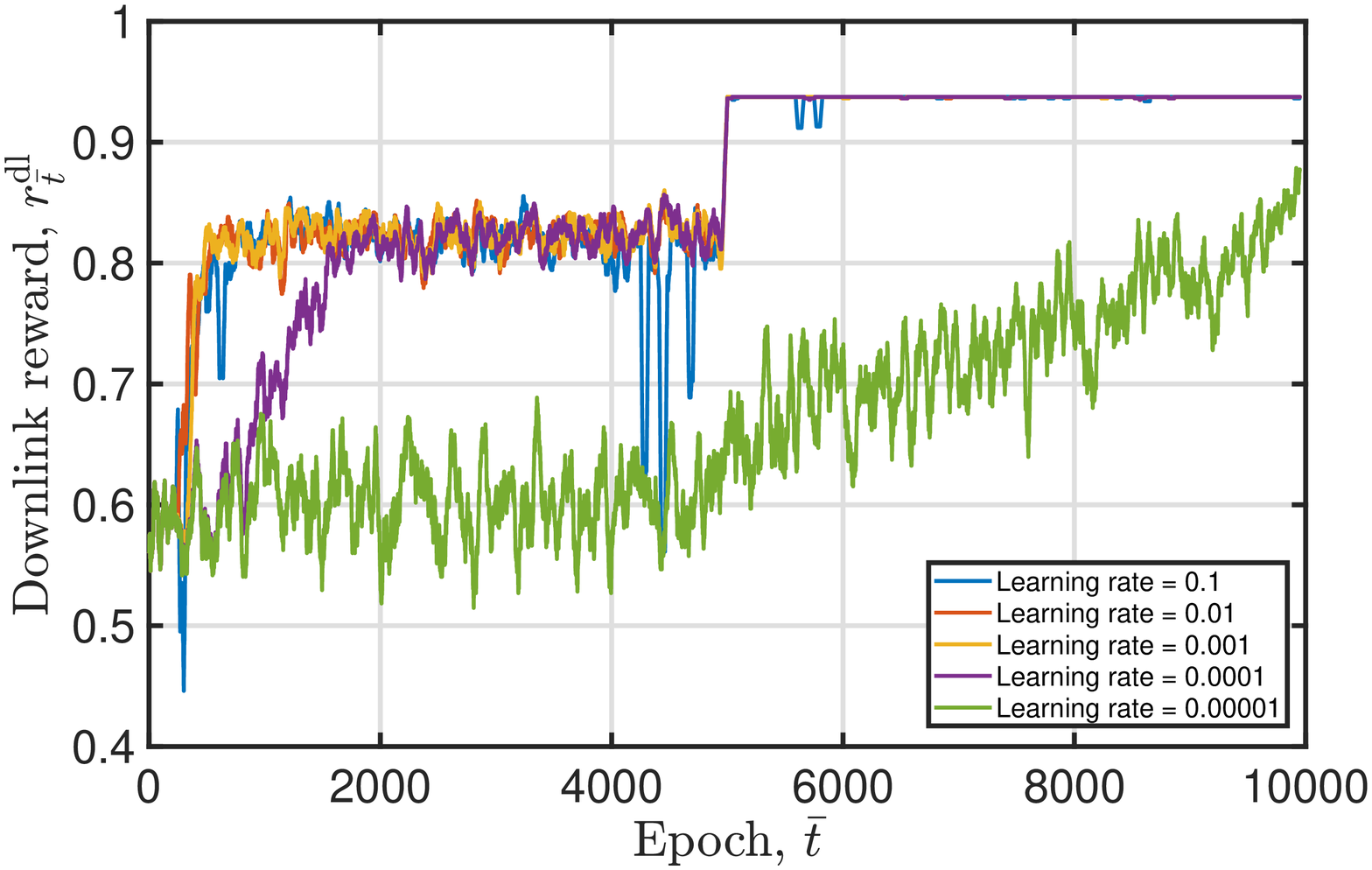}%
\label{fig_forth_case_lr}}
\caption{The impact of learning rates $l_r^{\rm ul}$ and $l_r^{\rm dl}$ on the convergence performance of the proposed algorithm.}
\label{fig_converge_learning_rate}
\end{figure*}

At last, we verify the superiority of the proposed algorithm by comparing it with other comparison algorithms. Particularly, we plot the achieved objective function values of all comparison algorithms under varying number of mobile users $N \in \{8, 12, 16, 20\}$ in Fig. \ref{fig_utility_user_number}.
Before the evaluation, the proposed algorithm and the other two action quantization algorithms have been trained with $10000$ independent wireless channel realizations, and their downlink and uplink action quantization policies have converged. This is reasonable because we are more interested in the long-term operation performance for field deployment. Besides, we let the service ability of an AP $\tilde M$ vary with $N$ with the $(N, \tilde M)$ pair being $(8, 3)$, $(12, 5)$, $(16, 6)$, and $(20, 7)$.
\begin{figure}[!t]
\centering
\includegraphics[width=3.2in]{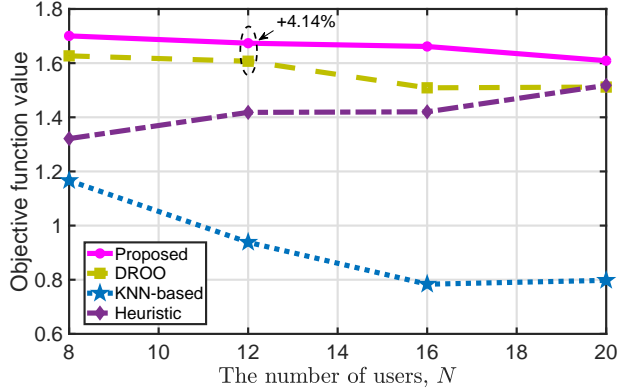}
\caption{Comparison of obtained objective function values of all comparison algorithms.}
\label{fig_utility_user_number}
\end{figure}

We have the following observations from this figure:
i) the proposed algorithm achieves the greatest objective function value. For the DROO algorithm, it gains a smaller objective function value than the proposed algorithm; for example, the achieved objective function value of DROO is $4.14\%$ less than that of the proposed algorithm. For the KNN-based algorithm, it obtains the smallest objective function value because it offers the smallest diversity in the produced uplink and downlink association action set;
ii) except for heuristic algorithm, the achieved objective function values of the other comparison algorithms decrease with the number of users owing to the increasing total power consumption. For the heuristic algorithm, its obtained objective function value increases with $N$ mainly because more users can successfully access to APs.

\section{Conclusion}
This paper investigated the problem of enhancing VR visual experiences for mobile users and formulated the problem as a sequence-dependent problem aiming at maximizing users' feeling of presence in VR environments while minimizing the total power consumption of users' HMDs.
This problem was confirmed to be a mixed-integer and non-convex optimization problem, the solution of which also needed accurate users' tracking information.
To solve this problem effectively, we developed a parallel ESN learning method to predict users' tracking information, with which a DRL-based optimization algorithm was proposed.
Specifically, this algorithm first decomposed the formulated problem into an association subproblem and a power control subproblem. Then, a DNN joint with an action quantization scheme was implemented as a scalable solution that learnt association variables from experience. Next, the power control subproblem with an SDR scheme being explored to tackle its non-convexity was leveraged to criticize the association variables.
Finally, simulation results were provided to verify the accuracy of the learning method and showed that the proposed algorithm could improve the energy efficiency by at least $4.14\%$ compared with various benchmark algorithms.

\appendix

\subsection{Proof of Lemma 1}
For any user $i \in {\mathcal U}$, suppose we are provided with a sequence of $Q$ desired input-output pairs $\{(\bm x_{i(t-Q)}, {\bm y}_{i(t-Q+1)}), \ldots, (\bm x_{i(t-1)}, {\bm y}_{it})\}$. With the input-output pairs, generate the hidden matrix ${\bm X_{it}} = \left[ {\begin{array}{*{20}{c}}
\begin{array}{l}
{\bm x_{i(t - 1)}}\\
{\bm s_{i(t - 1)}}
\end{array}& \cdots &\begin{array}{l}
{\bm x_{i(t - Q)}}\\
{\bm s_{i(t - Q)}}
\end{array}
\end{array}} \right]$ and the corresponding target location matrix ${\bm Y}_{it} = [\bm y_{it}^{\rm T}; \ldots; \bm y_{i(t-Q+1)}^{\rm T}]$ at time slot $t$. We next introduce an auxiliary matrix $\bm U = \bm X^{\rm T} \bm W \in {\mathbb R}^{Q \times N_o}$, wherein we lighten the notation $\bm X_{it}$ for $\bm X$. According to the Lagrange dual decomposition method, we can rewrite (\ref{eq:least_square_prob}) as follows
\begin{equation}\label{eq:lagrangian_dual}
\begin{array}{l}
\frac{1}{Q}\mathop {\rm minimize }\limits_{\bm W, \bm U} \quad \left\{ { {l(\bm X^{\rm T} \bm W)}  + \xi Qr(\bm W)} \right . +  \left . \bm A \odot \bm U  -  \bm A \odot \bm {X^{\rm T}W} \right \} \\
 = \frac{1}{Q}\mathop {\inf }\limits_{\bm W} \left\{ { -  \sum\limits_{n = 1}^{N_o}{[\bm A \bm z_n]^{\rm T} {\bm X^{\rm T}\bm W \bm z_n}}  + \xi Q r(\bm W)} \right\} + \\
\qquad \qquad \frac{1}{Q}\mathop {\inf }\limits_{\bm U} \left\{ { {l(\bm U)}  +  \sum\limits_{n=1}^{N_o}{[\bm A \bm z_n]^{\rm T} \bm U \bm z_n}} \right\}\\
 =  - {\xi} \mathop {\sup }\limits_{\bm W} \left\{ {\frac{1}{{\xi Q}} \sum\limits_{n = 1}^{N_o}{[\bm A \bm z_n]^{\rm T} \bm X^{\rm T} \bm W \bm z_n}  - {r(\bm W)}} \right\} - \\
\qquad \qquad \frac{1}{Q} \sum\limits_{j=1}^{J} \sum\limits_{m \in {\mathcal Q}_j} { \sum\limits_{n=1}^{N_o} {\mathop {\sup }\limits_{u_{mn}} \left\{ { -  a_{mn}u_{mn}  -  {l(u_{mn})} } \right\} } }  \\
 =  - \xi {r^{\star}}\left (\frac{1}{\xi Q} {\bm A^{\rm T} {\bm X^{\rm T}}} \right ) - \frac{1}{Q} \sum\limits_{j = 1}^{J} \sum\limits_{m \in {\mathcal Q}_j} {\sum\limits_{n=1}^{N_o} {l^{\star} (-a_{mn})} } \\
  \buildrel \Delta \over =  D(\bm A)
\end{array}
\end{equation}
where $\bm z_n \in {\mathbb R}^{N_o}$ is a column vector with the $n$-th element being one and all other elements being zero,
${\mathcal Q}_j$ is an index set including the indices of $Q$ data samples fed to slave VM $j$.

Let $\bar r(\bm C) =  {\frac{1}{\xi Q} \sum\limits_{n = 1}^{N_o}{\bm z_n^{\rm T} \bm C {\bm {Wz_n}}}  - {r(\bm W)}} $, where ${\bm C} = \frac{1}{{\xi Q}} \bm A^{\rm T} \bm X^{\rm T}$, and denote ${\bm W}^{\star}$ as the optimal solution to $\mathop {\sup }\limits_{\bm W} {\bar r(\bm C)}$. Then, calculate the derivative of ${\bar r(\bm C)}$ w.r.t ${\bm W}$,
\begin{equation}
\frac{{d \bar r(\bm C)}}{{d {\bm W}}} = \sum_{n=1}^{N_o} {{\bm C}_n {\bm z}_n^{\rm T} } - 2 {\bm W}
\end{equation}
where ${\bm C}_n = {\bm C}^{\rm T} {\bm z}_n$.

As ${\bm W \in {\mathbb R}^{{(N_i + N_r)} \times N_o}}$, the necessary and sufficient condition for obtaining ${\bm W}^{\star}$ is to enforce $\frac{{d \bar r (\bm C)}}{{d {\bm W}^{\star}}} = 0$. Then, we have
\begin{equation}\label{eq:W_star}
{\bm W}^{\star} = \frac{1}{2} \sum_{n=1}^{N_o} {{\bm C}_n {\bm z}_n^{\rm T} }
\end{equation}

By substituting ({\ref{eq:W_star}}) into $r^{\star}({\bm C})$, we can obtain (\ref{eq:r_star}).

Similarly, denote $u_{mn}^{\star}$ for any $m \in \{1, 2, \ldots, Q\}$ and $n \in \{1, 2, \ldots, N_o\}$ as the optimal solution to $l^{\star}(-a_{mn})$. As ${\bm U} \in {\mathbb R}^{Q \times N_o}$, the necessary and sufficient condition for $u_{mn}^{\star}$ is to execute $\frac{{d l^{\star}(-a_{mn})}}{{d {u_{mn}^{\star}}}} = -a_{mn} - u_{mn}^{\star} + y_{mn} = 0$. By substituting $u_{mn}^{\star}$ into $l^{\star}(-a_{mn})$, we can obtain (\ref{eq:l_star}). This completes the proof.


\ifCLASSOPTIONcaptionsoff
  \newpage
\fi




%
\bibliographystyle{IEEEtran}
\bibliography{FL_ES_VR}

\end{document}